%% file: equivalence_first-_and_second-order_formalism_electromagnetism.tex
\documentclass[a4paper,10pt,numbers=noenddot,twoside=on,headinclude=true,footinclude=false,headsepline=true]{scrartcl}

\usepackage[utf8]{inputenc}
\usepackage[T1]{fontenc}
\usepackage{textcomp}
\usepackage[english]{babel}
\usepackage{color}
\usepackage{amssymb}
\usepackage{amsmath}
\usepackage{amsfonts}
\usepackage{amsopn}
\usepackage{amsbsy}
\usepackage{dsfont}
\usepackage[overload,ntheorem]{empheq}
\usepackage{enumerate}
\usepackage[plainpages=false,pdfpagelabels,breaklinks=true,pdfborderstyle={/S/U/W 1.2}]{hyperref}
\usepackage{graphics}
\usepackage{units}
\usepackage[cal=boondoxo,bb=boondox]{mathalfa}
\numberwithin{equation}{section} 
\numberwithin{table}{section} 
\numberwithin{figure}{section} 
\usepackage[amsmath,thmmarks,thref,hyperref]{ntheorem}

\theoremstyle{plain}
\newtheorem{theorem}{Theorem}[section]

\newtheorem{lemma}[theorem]{Lemma}
\newtheorem{conjecture}[theorem]{Conjecture}
\newtheorem{corollary}[theorem]{Corollary}
\newtheorem{proposition}[theorem]{Proposition}
\newtheorem{assumption}[theorem]{Assumption}

\theorembodyfont{\upshape}
\newtheorem{remark}[theorem]{Remark}

\theoremstyle{nonumberplain}

\theoremsymbol{\ensuremath{_\Box}}
\newtheorem{proof}{Proof}
\usepackage[style=alphabetic,
	citestyle=alphabetic,
	firstinits=true,
	backend=biber,
	hyperref=auto,
	sorting=nyt,
	maxcitenames=3,
	maxbibnames=99,
	minnames=3,
	arxiv=pdf,
	url=false,
	isbn=false,
	dateabbrev=true,
	datezeros=true,
	bibencoding=utf8]{biblatex}
\newbibmacro{string+doi}[1]{%
	\iffieldundef{doi}{#1}{\href{http://dx.doi.org/\thefield{doi}}{#1}}}
\DeclareFieldFormat
	{title}{\usebibmacro{string+doi}{\mkbibemph{#1}}\isdot}
\DeclareFieldFormat
	[article,inbook,incollection,inproceedings,patent,thesis,unpublished]
	{title}{\usebibmacro{string+doi}{\mkbibemph{#1}}\isdot}
\DeclareFieldFormat
	[article,inbook,incollection,inproceedings,patent,thesis,unpublished]
	{pages}{#1}
\DeclareFieldFormat
	[article]
	{journaltitle}{#1\isdot}
\DeclareFieldFormat
	[article]
	{volume}{\textbf{#1}}
\DefineBibliographyStrings{english}{pages={}}
\DeclareBibliographyDriver{article}{%
	\usebibmacro{bibindex}%
	\usebibmacro{begentry}%
	\usebibmacro{author/translator+others}%
	\setunit{\labelnamepunct}\newblock
	\usebibmacro{title}%
	\usebibmacro{journal+issuetitle}%
	\setunit*{\addcomma\space}
	\printfield{pages}
	\setunit*{\addcomma\space}
	\printfield{year}
	\usebibmacro{finentry}%
}

\newbibmacro*{journal+issuetitle}{%
	\usebibmacro{journal}%
	\setunit*{\addspace}%
	\iffieldundef{series}
	{}
	{\newunit
		\printfield{series}%
		\setunit{\addspace}}%
	\iffieldundef{volume}
		{}
		{\printfield{volume}%
		\setunit{\addspace}}%
	\setunit*{\addcolon\space}%
	\usebibmacro{issue}%
	\newunit}
\usepackage[scaled]{beramono}
\usepackage{helvet}
\usepackage[charter]{mathdesign} % math font has ugly \mathcals
\SetMathAlphabet{\mathcal}{normal}{OMS}{cmsy}{m}{n} % fixes ugly \mathcals
\SetMathAlphabet{\mathcal}{bold}{OMS}{cmsy}{m}{n} % fixes ugly \mathcals
\providecommand{\ie}{i.~e.~}
\providecommand{\eg}{e.~g.~}
\providecommand{\cf}{cf.~}

\providecommand{\R}{\mathbb{R}}

\providecommand{\C}{\mathbb{C}}
\renewcommand{\C}{\mathbb{C}}

\providecommand{\T}{\mathbb{T}}
\renewcommand{\T}{\mathbb{T}}
\providecommand{\N}{\mathbb{N}}
\providecommand{\Z}{\mathbb{Z}}
\providecommand{\ii}{\mathrm{i}}
\providecommand{\e}{\mathrm{e}}
\renewcommand{\Re}{\mathrm{Re} \,}

\providecommand{\Hil}{\mathcal{H}}
\providecommand{\eps}{\varepsilon}
\providecommand{\Cont}{\mathcal{C}}

\providecommand{\ker}{\mathrm{ker} \, }
\providecommand{\ran}{\mathrm{ran} \, }

\providecommand{\ker}{\mathrm{ker} \,}

\providecommand{\dd}{\mathrm{d}}
\providecommand{\id}{\mathds{1}}

\providecommand{\Fourier}{\mathcal{F}}

\providecommand{\abs}[1]{\left \lvert #1 \right \rvert}

\providecommand{\babs}[1]{\bigl \lvert #1 \bigr \rvert}

\providecommand{\norm}[1]{\left \lVert #1 \right \rVert}
\providecommand{\snorm}[1]{\lVert #1 \rVert}
\providecommand{\bnorm}[1]{\bigl \lVert #1 \bigr \rVert}

\providecommand{\scpro}[2]{\left \langle #1 , #2 \right \rangle}
\providecommand{\sscpro}[2]{\langle #1 , #2 \rangle}
\providecommand{\bscpro}[2]{\bigl \langle #1 , #2 \bigr \rangle}
\providecommand{\Bscpro}[2]{\Bigl \langle #1 , #2 \Bigr \rangle}

\providecommand{\sket}[1]{\vert #1 \rangle}
\providecommand{\bket}[1]{\bigl \vert #1 \bigr \rangle}

\providecommand{\bra}[1]{\left \langle #1 \right \vert}
\providecommand{\sbra}[1]{\langle #1 \vert}
\providecommand{\bbra}[1]{\bigl \langle #1 \bigr \vert}

\providecommand{\Maux}{M^{\mathrm{aux}}}
\providecommand{\Div}{\mathrm{Div}}
\providecommand{\Rot}{\mathrm{Rot}}

\providecommand{\domain}{\mathcal{D}}
\providecommand{\WS}{\mathbb{M}}
\providecommand{\BZ}{\WS^*}
\providecommand{\specrel}{\sigma_{\mathrm{rel}}}
\bibliography{/Users/max/Library/texmf/tex/latex/max/bibliography}
\title{Equivalence of Electric, Magnetic and Electromagnetic Chern Numbers for Topological Photonic Crystals}
\author{Giuseppe De Nittis${}^1$ \& Max Lein${}^2$}

\begin{document}

\maketitle
\vspace{-9mm}
\begin{center}
	${}^1$ Facultad de Matemáticas \& Instituto de Física, 
	Pontificia Universidad Católica de Chile \linebreak
	Avenida Vicuña Mackenna 4860, 
	Santiago, 
	Chile \linebreak
	{\footnotesize \href{mailto:gidenittis@mat.uc.cl}{\texttt{gidenittis@mat.uc.cl}}}
	\medskip
	\\
	${}^2$ Advanced Institute of Materials Research,  
	Tohoku University \linebreak
	2-1-1 Katahira, Aoba-ku, 
	Sendai, 980-8577, 
	Japan \linebreak
	{\footnotesize \href{mailto:max.lein@tohoku.ac.jp}{\texttt{max.lein@tohoku.ac.jp}}}
\end{center}
\begin{abstract}
	Haldane \cite{Raghu_Haldane:quantum_Hall_effect_photonic_crystals:2008} predicted an analog of the Integer Quantum Hall Effect in gyrotropic photonic crystals, where the net number of electromagnetic edge modes moving left-to-right is given by a bulk Chern number. His prediction — topological effects are \emph{bona fide} wave and not quantum phenomena — has been confirmed in a number of experiments \cite{Wang_et_al:unidirectional_backscattering_photonic_crystal:2009}. However, theoretical physicists have tacitly used three different definitions for the bulk Chern numbers that enter the bulk-edge correspondence — on the basis of \emph{electromagnetic} Bloch functions, \emph{electric} Bloch functions and \emph{magnetic} Bloch functions. We use vector bundle theoretic arguments to prove that in media such as those considered by Haldane these three potentially different Chern numbers necessarily agree with one another, and consequently, any one of them can be used in Haldane's photonic bulk-edge correspondence. 
\end{abstract}
\\
{\scriptsize \textbf{PACS 2010:} 41.20.Jb, 42.70.Qs, 78.20.-e}

\newpage
\tableofcontents

%%% begin content %%% (fold)
\input{section_1}
\input{section_2}
\input{section_3}
% \input{section_4}
% \input{section_5}
\input{appendix}
%%% end content %%% (end)

\printbibliography

\end{document}

%% file: section_1.tex
%!TEX root = /Users/max/Dropbox/research/photonic crystals/active/equivalence first- and second-order formalism/paper/equivalence first- and second-order formalism electromagnetism.tex
\section{Introduction} % (fold)
\label{intro}
Photonic crystals are to electromagnetic waves what crystalline solids are to an electron, and therefore it is not surprising that periodic electromagnetic media exhibit analogs to various phenomena from condensed matter physics. One particularly intriguing example of such a quantum-wave analogy is the \emph{“Quantum Hall Effect for light”} in media with broken time-reversal symmetry. While nowadays whole sub communities are working on topological phenomena in classical waves, Raghu's and Haldane's seminal idea \cite{Raghu_Haldane:quantum_Hall_effect_photonic_crystals:2008} was initially met with a lot of skepticism. That is until researchers at MIT confirmed Haldane's prediction in experiment \cite{Wang_et_al:unidirectional_backscattering_photonic_crystal:2009,Ozawa_et_al:review_topological_photonics:2018}. 

Just like in case of the well-known Integer Quantum Hall Effect \cite{von_Klitzing_Dorda_Pepper:quantum_hall_effect:1980,von_Klitzing:quantum_Hall_effect:2004,Thouless_Kohmoto_Nightingale_Den_Nijs:quantized_hall_conductance:1982}, Haldane proposed to explain the existence and robustness of unidirectional, back-scattering-free edge modes that are at the heart of this phenomenon by means of a \emph{bulk-edge correspondence} \cite{Hatsugai:Chern_number_edge_states:1993,Hatsugai:edge_states_Riemann_surface:1993}: 
\begin{conjecture}[Raghu and Haldane's Photonic Bulk-Edge Correspondence \cite{Raghu_Haldane:quantum_Hall_effect_photonic_crystals:2008}]\label{intro:conjecture:photonic_bulk_edge_correspondence}
	In a two-dimensional photonic crystals with boundary the difference of the number of left- and right-moving boundary modes in bulk band gaps is a topologically protected quantity; it equals the bulk Chern number associated to the positive frequency bands below the bulk band gap. 
\end{conjecture}
Chern numbers are topological invariants that are associated to families of frequency bands, the \emph{relevant bands}, that are separated from the others by a \emph{spectral gap} (the Gap Condition~\ref{periodic:assumption:gap_condition} makes this precise); they cannot change under continuous, gap-preserving transformations, which explains their robustness under rather strong perturbations. 

Our ultimate goal is to furnish a proof to Haldane's conjecture, and our earlier publications \cite{DeNittis_Lein:adiabatic_periodic_Maxwell_PsiDO:2013,DeNittis_Lein:sapt_photonic_crystals:2013,DeNittis_Lein:ray_optics_photonic_crystals:2014,DeNittis_Lein:symmetries_Maxwell:2014,DeNittis_Lein:Schroedinger_formalism_classical_waves:2017,DeNittis_Lein:symmetries_electromagnetism:2017} all systematically work towards this goal. 

This paper addresses one aspect of this endeavor: Maxwell's equations~\eqref{setup:eqn:Maxwell_equations} for the \emph{electromagnetic} field are \emph{first} order in time and space. When bi-anisotropic coupling between electric and magnetic fields is absent, the dynamical Maxwell equations~\eqref{setup:eqn:Maxwell_equations:dynamics} are block-\emph{off}diagonal in $(\mathbf{E},\mathbf{H})$, and it is possible and indeed, sometimes preferable to work with \emph{second}-order wave equations for the \emph{electric} or \emph{magnetic} fields instead (equations~\eqref{setup:eqn:wave_equation_electric} and \eqref{setup:eqn:wave_equation_magnetic}, respectively). That makes the second-order formalism appealing for numerical schemes, because it is inherently more efficient to work with $\C^3$-valued vector fields rather than $\C^6$-valued vector fields. Moreover, for certain geometries (\cf \eg \cite[Section~2.4]{DeNittis_Lein:symmetries_Maxwell:2014}) the second-order formalism allows us to describe all of the physics via two \emph{scalar} equations as opposed to $\C^6$-valued vector fields. In view of this, it is not surprising that in many applications the Chern numbers that enter Haldane's photonic bulk-boundary conjecture are computed on the basis of the electric or magnetic field alone (see \eg \cite{Joannopoulos_Johnson_Winn_Meade:photonic_crystals:2008,Wang_et_al:edge_modes_photonic_crystal:2008,Wang_et_al:unidirectional_backscattering_photonic_crystal:2009}); we will denote those electric and magnetic Chern numbers with $\mathrm{Ch}^E$ and $\mathrm{Ch}^H$, respectively. 

On the other hand, analogies to quantum mechanics are most readily apparent in the \emph{first}-order formalism. Indeed, Haldane's own work uses the \emph{electromagnetic} field to compute Chern numbers $\mathrm{Ch}^{EH}$. And we have argued in \cite[Section~5.2.1]{DeNittis_Lein:symmetries_electromagnetism:2017} that conceptually speaking, this is the right quantity to start with in a first-principles approach. In the same place, we also raised the question that we will answer in the affirmative here: 
\begin{theorem}[Electric, magnetic and electromagnetic bulk Chern numbers agree]\label{intro:thm:Chern_numbers_agree}~\\
	Suppose the material weights which describe the medium satisfy Assumption~\ref{periodic:assumption:periodic_material_weights}, \ie the medium is periodic, lossless, not bi-anisotropic and has positive index. 
	Then the electric, magnetic and electromagnetic Chern numbers associated to any family of frequency bands satisfying the Gap Condition~\ref{periodic:assumption:gap_condition} agree, 
	\begin{align*}
		\mathrm{Ch}^E = \mathrm{Ch}^H = \mathrm{Ch}^{EH}
		. 
	\end{align*}
\end{theorem}
While the validity of this result is often tacitly assumed in the literature, to the best of our knowledge there is no publication that actually tries to \emph{derive} this from first principles. Indeed, we rephrase the problem in the — for physicists — more abstract language of vector bundles and then the proof is straightforward. In contrast, a direct verification of this Theorem by computing the electric/magnetic/electromagnetic Chern number from the electric/magnetic/electromagnetic Berry curvature seems unfeasible — even in the simplest case of a single, non-degenerate band. 
\begin{remark}
	The assumption that the medium is not bi-anisotropic is crucial. For otherwise, the electric and magnetic Chern numbers are not well-defined, and the question we set out to answer in this paper does not make sense. Indeed, when the Maxwell equations contain a non-zero bi-anisotropic coupling term, then electric and magnetic fields do \emph{not} decouple in the second-order formalism. Even block-diagonalizing the first-order equations in the $(\mathbf{E},\mathbf{H})$ splitting is usually problematic as that necessarily mixes positive and negative frequency states; for a more in-depth discussion we refer to \cite[Section~5.2.2]{DeNittis_Lein:symmetries_electromagnetism:2017}. 
\end{remark}
Strictly speaking, we only provide a proof of this statement for a three-dimensional photonic crystal whereas the Quantum Hall Effect for light occurs in quasi-two-dimensional media. However, none of the arguments depend on the dimensionality and can be readily adapted to lower- or higher-dimensional photonic crystals. Moreover, it should be possible to go from periodic to include randomness, but we shall not do this here. 

Our proof will consist of two steps: first, we will show the equivalence of the first- and second-order equations. While this is \emph{in principle} completely standard and is spelled out in the literature (including \cite{Wilcox:scattering_theory_classical_physics:1966,Reed_Simon:scattering_theory_wave_equations:1977,Figotin_Klein:localization_classical_waves_II:1997}), existing results do not cover the case we are interested in, media with broken time-reversal symmetry. Here, our only original contribution is to start with the correct, physically meaningful \emph{first}-order Maxwell equations that we have derived in  \cite[Section~2]{DeNittis_Lein:Schroedinger_formalism_classical_waves:2017}. 

The second step is to use the language of vector bundles, and interpret the maps $\imath^{E,H}(k) : \varphi_n^{E,H}(k) \mapsto \bigl ( \varphi_n^E(k) , \varphi_n^H(k) \bigr )$, defined in Section~\ref{periodic:frequency_bands:equivalence_spectra} below, which reconstruct electromagnetic Bloch functions from the electric or magnetic field component alone. These maps can be interpreted as vector bundle isomorphisms between an electric, magnetic and electromagnetic Bloch bundle. Up to isomorphism these are characterized by Chern numbers, and since these three bundles are isomorphic, their Chern numbers necessarily agree. 
% section introduction (end)

%% file: section_2.tex
%!TEX root = /Users/max/Dropbox/research/photonic crystals/active/equivalence first- and second-order formalism/paper/equivalence first- and second-order formalism electromagnetism.tex
\section{Two equivalent mathematical descriptions of electromagnetism} % (fold)
\label{setup}
Electromagnetic, electric and magnetic Chern numbers arise naturally, depending on the equation of motion one starts with. The purpose of this section is to introduce three equivalent equations which govern electromagnetic waves propagating in certain linear media. The properties of the medium are phenomenologically described by the electric permittivity $\eps \in L^{\infty} \bigl ( \R^3 \, , \, \mathrm{Mat}_{\C}(3) \bigr )$ and the magnetic permeability $\mu \in L^{\infty} \bigl ( \R^3 \, , \, \mathrm{Mat}_{\C}(3) \bigr )$; collectively, we will refer to 
\begin{align*}
	W = \left ( 
	\begin{matrix}
		\eps & 0 \\
		0 & \mu \\
	\end{matrix}
	\right ) 
	\in L^{\infty} \bigl ( \R^3 \, , \, \mathrm{Mat}_{\C}(6) \bigr )
\end{align*}
as the \emph{material weights}, and throughout this article, we will impose the following 
\begin{assumption}[Material weights]\label{setup:assumption:material_weights}
	\begin{enumerate}[(a)]
		\item The medium is \emph{lossless}, \ie $W(x) = W(x)^*$ takes values in the \emph{hermitian} matrices. 
		\item The medium is \emph{not a negative index material}, \ie there exist positive constants $C \geq c > 0$ so that $c \, \id \leq W \leq C \, \id$ holds. 
		\item The medium has no bianisotropy, \ie the block-offdiagonal terms of $W$ vanish. 
	\end{enumerate}
\end{assumption}
Under these conditions (first-order) Maxwell's equations and the two (second-order) wave equations for electric and magnetic fields admit an $L^2$-theory. In particular, they give rise to selfadjoint operators acting on \emph{complex} electromagnetic, electric and magnetic fields, respectively. In the end, though, all three descriptions are equivalent (\cf Theorem~\ref{setup:thm:equivalence_frameworks}). 

While all of this is standard, we would like to emphasize two important points before experts skip the remainder of this section: 
\begin{enumerate}[(1)]
	\item To break time-reversal symmetry in a non-bianisotroptic medium — a prerequisite to have topological phenomena — the material weights $W \neq \overline{W}$ have to be complex (\cf \cite[Proposition~3.2 and Theorem~3.3]{DeNittis_Lein:symmetries_electromagnetism:2017}). In the context of this article, $W \neq \overline{W}$ is a necessary condition to have non-zero Chern numbers. For such media physically meaningful Maxwell equations have been derived only recently in \cite[Section~2]{DeNittis_Lein:Schroedinger_formalism_classical_waves:2017}. Indeed, only the physically meaningful Maxwell equations are equivalent to either of the wave equations. 
	\item The maps $\imath^{E,H}$ defined in equations~\eqref{setup:eqn:inclusion} below that enter as an auxiliary quantity in this section will play a pivotal rôle in the proof of our main result, Theorem~\ref{intro:thm:Chern_numbers_agree}. 
\end{enumerate}

\subsection{First-order formalism: Maxwell's equations in matter} % (fold)
\label{setup:first_order}
Under the above assumptions, we can give rigorous meaning to \emph{Maxwell's equations in matter}
\begin{subequations}\label{setup:eqn:Maxwell_equations}
	\begin{align}
		\left (
		\begin{matrix}
			\eps & 0 \\
			0 & \mu \\
		\end{matrix}
		\right ) \, \frac{\partial}{\partial t} \left (
		\begin{matrix}
			\psi^E(t) \\
			\psi^H(t) \\
		\end{matrix}
		\right ) &= \left (
		\begin{matrix}
			+ \nabla \times \psi^H(t) \\
			- \nabla \times \psi^E(t) \\
		\end{matrix}
		\right )
		&&
		\mbox{(dynamical equation)}
		\label{setup:eqn:Maxwell_equations:dynamics}
		\\
		\left (
		\begin{matrix}
			\nabla \cdot \eps \psi^E(t) \\
			\nabla \cdot \mu \psi^H(t) \\
		\end{matrix}
		\right ) &= \left (
		\begin{matrix}
			0 \\
			0 \\
		\end{matrix}
		\right )
		&&
		\mbox{(constraint equation)}
		\label{setup:eqn:Maxwell_equations:constraint}
		\\
		\left (
		\begin{matrix}
			\psi^E(t_0) \\
			\psi^H(t_0) \\
		\end{matrix}
		\right ) &= \left (
		\begin{matrix}
			\phi^E \\
			\phi^H \\
		\end{matrix}
		\right )
		&&
		\mbox{(initial condition)}
	\end{align}
\end{subequations}
defined on the vector space 
\begin{align}
	\Hil = \Bigl \{ \Psi \in L^2(\R^3,\C^6) \; \; \big \vert \; \; \Psi \mbox{ is a $\omega \geq 0$ frequency wave} \Bigr \} 
\end{align}
composed of \emph{complex} non-negative frequency waves; we will give a mathematically precise definition of $\Hil$ in equation~\eqref{setup:eqn:first_order_Helmholtz_splitting} below. For the benefit of the reader, we present enough details here to give rigorous meaning to equations~\eqref{setup:eqn:Maxwell_equations}; a detailed derivation with additional explanations can be found in \cite[Sections~2–3]{DeNittis_Lein:Schroedinger_formalism_classical_waves:2017}.

\subsubsection{Representing real-valued electromagnetic waves as complex $\omega \geq 0$ waves} % (fold)
\label{setup:first_order:representation_complex_waves}
The key idea of \cite{DeNittis_Lein:Schroedinger_formalism_classical_waves:2017} is to write a real electromagnetic field $(\mathbf{E},\mathbf{H}) = \Psi_+ + \Psi_-$ as the sum of two \emph{complex} waves $\Psi_{\pm}$ composed solely of non-negative ($+$) and non-positive ($-$) frequencies. As sources are absent and electromagnetic fields must be transversal, there are no zero frequency fields contributing to $\Psi_{\pm}$. Hence, we will call $\Psi_{\pm}$ the positive/negative frequency contribution. To ensure that their sum is real, $\Psi_+$ and 
\begin{align}
	\Psi_- = \overline{\Psi_+} 
	\label{setup:eqn:phase_locking_constraint}
\end{align}
are phase locked. Put another way, $\Psi_+$ and $\Psi_-$ are \emph{not} independent degrees of freedom, and we may pick one of the two — typically $\Psi_+$ — to describe the real wave $(\mathbf{E},\mathbf{H}) = 2 \Re \Psi_+$. In this sense, $\Psi_+$ is a complex wave \emph{representing} the real electromagnetic field $(\mathbf{E},\mathbf{H})$. As we shall see below, this correspondence $(\mathbf{E},\mathbf{H}) \leftrightarrow \Psi_+$ is one-to-one (\cf Proposition~\ref{setup:prop:complex_representation_real_fields}). 

When the weights $W \neq \overline{W}$ are complex, then $\Psi_+$ and $\Psi_-$ evolve according to \emph{different} Maxwell equations: in order to ensure the phase locking condition~\eqref{setup:eqn:phase_locking_constraint}, the Maxwell equations for $\Psi_-$ involve the complex conjugate weights. Thus, the restriction of \eqref{setup:eqn:Maxwell_equations} to non-negative frequencies is crucial as the negative frequency solutions to \eqref{setup:eqn:Maxwell_equations} \emph{sans} frequency restriction are unphysical. 

One of the main points was to show that Maxwell's equations~\eqref{setup:eqn:Maxwell_equations} of lossless positive index media can be recast in the form of a Schrödinger equation
\begin{align}
	\ii \partial_t \Psi(t) = M \Psi(t) 
	, 
	&&
	\Psi(t_0) = \Phi
	. 
	\label{setup:eqn:Schroedinger_form_Maxwell_equations}
\end{align}
To rigorously define the \emph{Maxwell operator} $M$, we first multiply both sides of \eqref{setup:eqn:Maxwell_equations:dynamics} with $\ii \, W^{-1}$, which yields $\ii \partial_t \Psi(t) = \Maux \, \Psi(t)$ where the \emph{auxiliary Maxwell operator}
\begin{align}
	\Maux := W^{-1} \, \Rot := \left (
	\begin{matrix}
		\eps^{-1} & 0 \\
		0 & \mu^{-1} \\
	\end{matrix}
	\right ) \, \left (
	\begin{matrix}
		0 & + \ii \nabla^{\times} \\
		- \ii \nabla^{\times} & 0 \\
	\end{matrix}
	\right )
	\label{setup:eqn:auxiliary_Maxwell_operator}
\end{align}
lacks the restriction to non-negative frequencies. Endowed with the domain $\domain(\Maux) := \domain(\Rot)$ of the free Maxwell operator (made explicit in \cite[equation~(15)]{DeNittis_Lein:adiabatic_periodic_Maxwell_PsiDO:2013}), $\Maux$ defines a closed operator on the \emph{Banach} space $L^2(\R^3,\C^6)$. Once we equip this $L^2$-space with the weighted \emph{energy scalar product} 
\begin{align}
	\scpro{\Phi}{\Psi}_W := \bscpro{\Phi}{W \, \Psi}_{L^2(\R^3,\C^6)}
	, 
	\label{setup:eqn:energy_scalar_product}
\end{align}
we obtain the Hilbert space $L^2_W(\R^3,\C^6)$; note that due to our assumptions on the material weights, $L^2_W(\R^3,\C^6)$ agrees with the ordinary, unweighted $L^2(\R^3,\C^6)$ as Banach spaces. 

On this weighted $L^2$-space $\Maux = \bigl ( \Maux \bigr )^{\ast_W}$ is selfadjoint \cite[Proposition~6.2]{DeNittis_Lein:Schroedinger_formalism_classical_waves:2017}, and the spectral projections 
\begin{subequations}
	\begin{align}
		P_+ :& \negmedspace = 1_{(0,\infty)}(\Maux) 
		\label{setup:eqn:positive_frequency_projection}
		\\
		P_0 :& \negmedspace = 1_{\{ 0 \}}(\Maux) 
		\label{setup:eqn:zero_frequency_projection}
	\end{align}
\end{subequations}
onto the positive and zero frequency contributions can be defined via functional calculus. The Hilbert space of complex $\omega \geq 0$ waves is the corresponding spectral subspace, 
\begin{align}
	\Hil :& \negmedspace = 1_{[0,\infty)}(\Maux) \bigl [ L^2_W(\R^3,\C^6) \bigr ]
	\notag \\
	&= \mathcal{J}_+ \oplus \mathcal{G} 
	:= P_+ \bigl [ L^2_W(\R^3,\C^6) \bigr ] \oplus P_0 \bigl [ L^2_W(\R^3,\C^6) \bigr ] 
	, 
	\label{setup:eqn:first_order_Helmholtz_splitting}
\end{align}
that we further divide into positive and zero frequency components. This \emph{Helmholtz splitting} (\cf \cite[Section~3.2.1]{DeNittis_Lein:Schroedinger_formalism_classical_waves:2017}) is conceptually important since longitudinal gradient fields that make up $\mathcal{G}$ are $\scpro{\, \cdot \,}{\, \cdot \,}_W$-orthogonal to the positive frequency subspace $\mathcal{J}_+$ — so that its elements therefore automatically satisfy the transversality constraint~\eqref{setup:eqn:Maxwell_equations:constraint} in the weak sense. 

Any real electromagnetic field $(\mathbf{E},\mathbf{H}) \in L^2(\R^3,\R^6)$ has a unique representative $\Phi := Q (\mathbf{E},\mathbf{H})$ via the map 
\begin{align}
	Q := \bigl ( P_+ + \tfrac{1}{2} P_0 \bigr ) \big \vert_{L^2(\R^3,\R^6)} : L^2(\R^3,\R^6) \longrightarrow \Hil
	. 
	\label{setup:eqn:Q_complex_representative_map}
\end{align}
\begin{proposition}[{{\cite[Corollary~A.2]{DeNittis_Lein:Schroedinger_formalism_classical_waves:2017}}}]\label{setup:prop:complex_representation_real_fields}
	Suppose the material weights satisfy Assumption~\ref{setup:assumption:material_weights}. Then the maps 
	\begin{align*}
		Q &: L^2(\R^3,\R^6) \longrightarrow \Hil
		\\
		2 \Re \, Q &: L^2(\R^3,\R^6) \longrightarrow \Hil
	\end{align*}
	are injective. Hence, any real electromagnetic field $(\mathbf{E},\mathbf{H}) \in L^2(\R^3,\R^6)$ can be uniquely represented as a complex wave. 
\end{proposition}
Morally speaking, this one-to-one correspondence can be understood as follows: complexifying the \emph{real} vector space $L^2(\R^3,\C^6)$ “doubles” the degrees of freedom. One way to eliminate the superfluous elements is to restrict to complex waves of non-negative frequencies. A more careful analysis shows \cite[Lemma~2.5]{DeNittis_Lein:ray_optics_photonic_crystals:2014} that for real transversal fields, the map $P_+$ restricted to the subspace of real transversal fields is really a bijection onto $\mathcal{J}_+$. 
\begin{remark}
	When $W = \overline{W}$ is real, then a quick computation yields $Q \vert_{\ran Q}^{-1} = 2 \Re$ is just twice the real part. This computation also explains the presence of the factor $\nicefrac{1}{2}$ in euqation~\eqref{setup:eqn:Q_complex_representative_map}, it avoids $\omega = 0$ fields being counted twice. Even though we suspect this is also true for media with complex weights $W \neq \overline{W}$, we are presently not aware of a proof. For details we refer the interested reader to Section~3.2.2 and Appendix~A of \cite{DeNittis_Lein:Schroedinger_formalism_classical_waves:2017}. 
\end{remark}
%
% subsubsection representing_real_valued_electromagnetic_waves_as_complex_omega_geq_0_waves (end)

\subsubsection{The Schrödinger formalism of electromagnetism} % (fold)
\label{setup:first_order:Schroedinger_formalism}
The analog of the quantum Hamiltonian which enters the Schrödinger-type equation~\eqref{setup:eqn:Schroedinger_form_Maxwell_equations} is the Maxwell operator $M := \Maux \vert_{\omega \geq 0}$, that is obtained by restricting the \emph{auxiliary} Maxwell operator to the non-negative frequency Hilbert space $\Hil$. Hence, endowed with the obvious domain $\domain(M) := \domain(\Maux) \cap \Hil$ the Maxwell operator $M = M^{\ast_W}$ inherits the selfadjointness from its parent \cite[Lemma~B.2]{DeNittis_Lein:Schroedinger_formalism_classical_waves:2017}. 

Then a straightforward analysis shows that Maxwell's equations~\eqref{setup:eqn:Maxwell_equations} are equivalent to the Schrödinger equation~\eqref{setup:eqn:Schroedinger_form_Maxwell_equations}. 
\begin{theorem}[Equivalence of \eqref{setup:eqn:Maxwell_equations} and \eqref{setup:eqn:Schroedinger_form_Maxwell_equations} {{\cite[Theorem~3.4]{DeNittis_Lein:Schroedinger_formalism_classical_waves:2017}}}]\label{setup:thm:equivalence_frameworks}
	Suppose the material weights $W$ which describe the medium satisfy Assumption~\ref{setup:assumption:material_weights}. 
	Then the Maxwell equations~\eqref{setup:eqn:Maxwell_equations} are equivalent to the Schrödinger-type equation~\eqref{setup:eqn:Schroedinger_form_Maxwell_equations}. The real electromagnetic field $\bigl ( \mathbf{E}(t) , \mathbf{H}(t) \bigr ) = Q^{-1} \, \Psi(t)$ can be recovered from the inverse of the map $Q$. 
\end{theorem}
%
% subsubsection the_schrodinger_formalism_of_electromagnetism (end)
% subsection first_order_formalism_maxwell_s_equations_in_matter (end)

\subsection{Second-order formalism: wave equations for electric and magnetic fields} % (fold)
\label{setup:second_order}
Alternatively, we can square the Schrödinger-type equation~\eqref{setup:eqn:Schroedinger_form_Maxwell_equations} to obtain second-order wave equations; this has the advantage of yielding separate equations for electric 
\begin{subequations}\label{setup:eqn:wave_equation_electric}
	\begin{align}
		&\partial_t^2 \psi^E(t) + \eps^{-1} \, \nabla^{\times} \, \mu^{-1} \, \nabla^{\times} \psi^E(t) = 0 
		, 
		&&
		\mbox{(dynamical equation)}
		\label{setup:eqn:wave_equation_electric:dynamics}
		\\
		&\nabla \cdot \eps \psi(t) = 0 = \nabla \cdot \eps \, \partial_t \psi(t) 
		, 
		&&
		\mbox{(constraint equation)}
		\label{setup:eqn:wave_equation_electric:constraint}
		\\
		&\psi^E(t_0) = \phi^E
		, \; \, 
		\partial_t \psi^E(t_0) = + \eps^{-1} \, \nabla \times \phi^H 
		, 
		&&
		\mbox{(initial conditions)}
		\label{setup:eqn:wave_equation_electric:initial_conditions}
	\end{align}
\end{subequations}
and magnetic components, 
\begin{subequations}\label{setup:eqn:wave_equation_magnetic}
	\begin{align}
		&\partial_t^2 \psi^H(t) + \mu^{-1} \, \nabla^{\times} \, \eps^{-1} \, \nabla^{\times} \psi^H(t) = 0 
		, 
		&&
		\mbox{(dynamical equation)}
		\label{setup:eqn:wave_equation_magnetic:dynamics}
		\\
		&\nabla \cdot \mu \psi(t) = 0 = \nabla \cdot \mu \, \partial_t \psi(t) 
		, 
		&&
		\mbox{(constraint equation)}
		\label{setup:eqn:wave_equation_magnetic:constraint}
		\\
		&\psi^H(t_0) = \phi^H
		, \; \, 
		\partial_t \psi^H(t_0) = + \mu^{-1} \, \nabla \times \phi^H 
		. 
		&&
		\mbox{(initial conditions)}
		\label{setup:eqn:wave_equation_magnetic:initial_conditions}
	\end{align}
\end{subequations}
This is because the auxiliary Maxwell operator $\Maux = \left ( 
\begin{smallmatrix}
	0 & + \ii \eps^{-1} \nabla^{\times} \\
	- \ii \mu^{-1} \nabla^{\times} & 0 \\
\end{smallmatrix}
\right )$ is completely block-\emph{off}diagonal. Therefore, its square 
\begin{align}
	(\Maux)^2 = \left (
	\begin{matrix}
		M_{EE}^2 & 0 \\
		0 & M_{HH}^2 \\
	\end{matrix}
	\right )
	:= \left (
	\begin{matrix}
		\eps^{-1} \, \nabla^{\times} \, \mu^{-1} \, \nabla^{\times} & 0 \\
		0 & \mu^{-1} \, \nabla^{\times} \, \eps^{-1} \, \nabla^{\times} \\
	\end{matrix}
	\right )
	\label{setup:eqn:squared_auxiliary_Maxwell_operator}
\end{align}
is block-\emph{diagonal}, and the operators in the block-diagonals are those we are interested in. In fact, this equation encapsulates their precise mathematical definitions: the domains of $M_{EE}^2$ and $M_{HH}^2$ are electric and magnetic part of the domain of $(\Maux)^2 = M_{EE}^2 \oplus M_{HH}^2$, 
\begin{align*}
	\domain \bigl ( M_{EE}^2 \bigr ) :& \negmedspace = \domain \bigl ( (\Maux)^2 \bigr )^E 
	:= \Bigl \{ \psi^E \in L^2_{\eps}(\R^3,\C^3) \; \; \big \vert \; \; 
	\Psi = \left ( 
	\begin{smallmatrix}
		\psi^E \\
		\psi^H \\
	\end{smallmatrix}
	\right ) \in \domain \bigl ( (\Maux)^2 \bigr ) 
	\Bigr \} 
	, 
	\\
	\domain \bigl ( M_{HH}^2 \bigr ) :& \negmedspace = \domain \bigl ( (\Maux)^2 \bigr )^H 
	:= \Bigl \{ \psi^H \in L^2_{\mu}(\R^3,\C^3) \; \; \big \vert \; \; 
	\Psi = \left ( 
	\begin{smallmatrix}
		\psi^E \\
		\psi^H \\
	\end{smallmatrix}
	\right ) \in \domain \bigl ( (\Maux)^2 \bigr ) 
	\Bigr \} 
	. 
\end{align*}
\begin{remark}
	This definition via $\Maux$ imposes only the bare minimum of conditions on the domain; in principle, we \emph{could} attempt to define $M_{EE}^2$ and $M_{HH}^2$ directly without making any reference to the first-order operators, but that is actually more delicate if $\eps$ and $\mu$ are not $\Cont^1$-regular with bounded first-order derivatives. 
\end{remark}
Consequently, $M_{EE}^2$ and $M_{HH}^2$, seen as operators on the electric Hilbert space $L^2_{\eps}(\R^3,\C^3)$ and the magnetic Hilbert space $L^2_{\mu}(\R^3,\C^3)$, are selfadjoint. These Hilbert spaces are defined just as $L^2_W(\R^3,\C^6)$ by endowing the \emph{Banach space} $L^2(\R^3,\C^3)$ with the weighted scalar products 
\begin{align*}
	\bscpro{\phi^E}{\psi^E}_{\eps} :& \negmedspace = \bscpro{\phi^E}{\eps \psi^E}_{L^2(\R^3,\C^3)} 
	, 
	\\ 
	\bscpro{\phi^H}{\psi^H}_{\mu} :& \negmedspace = \bscpro{\phi^H}{\mu \psi^H}_{L^2(\R^3,\C^3)} 
	. 
\end{align*}
\begin{lemma}
	$M_{EE}^2 = (M_{EE}^2)^{\ast_{\eps}}$ and $M_{HH}^2 = (M_{HH}^2)^{\ast_{\mu}}$ are selfadjoint. 
\end{lemma}
While equations~\eqref{setup:eqn:wave_equation_electric} and \eqref{setup:eqn:wave_equation_magnetic} makes it seem as if electric and magnetic fields decouple, this is of course not the case: as second-order equations, we not only need to specify $\psi^E(t_0)$ or $\psi^H(t_0)$ but also the time-derivative 
\begin{subequations}\label{setup:eqn:time_derivative_intial_conditions_second_order}
	\begin{align}
		\partial_t \psi^E(t_0) &= + \eps^{-1} \, \nabla \times \psi^H(t_0)
		\label{setup:eqn:time_derivative_intial_conditions_second_order:electric}
		\\
		\partial_t \psi^H(t_0) &= - \mu^{-1} \, \nabla \times \psi^E(t_0)
		\label{setup:eqn:time_derivative_intial_conditions_second_order:magnetic}
	\end{align}
\end{subequations}
that evidently has to satisfy Maxwell's equations~\eqref{setup:eqn:Maxwell_equations}. Part and parcel is the assumption that $\Psi(t_0) = \bigl ( \psi^E(t_0) , \psi^H(t_0) \bigr ) \in \Hil$ is composed of positive frequencies. 

Lastly, akin to \eqref{setup:eqn:first_order_Helmholtz_splitting} let us introduce the Helmholtz decomposition for the electric and magnetic Hilbert spaces, 
\begin{subequations}\label{setup:eqn:second_order_Helmholtz_splitting}
	\begin{align}
		L^2_{\eps}(\R^3,\C^3) &= \bigl ( \ran \nabla \bigr )^{\perp_{\eps}} \oplus \bigl ( \ran \nabla \bigr )
		= \ker \, ( \nabla \cdot \eps ) \oplus \bigl ( \ran \nabla \bigr )
		, 
		\label{setup:eqn:second_order_Helmholtz_splitting:electric}
		\\
		L^2_{\mu}(\R^3,\C^3) &= \bigl ( \ran \nabla \bigr )^{\perp_{\mu}} \oplus \bigl ( \ran \nabla \bigr )
		= \ker \, ( \nabla \cdot \mu ) \oplus \bigl ( \ran \nabla \bigr )
		. 
		\label{setup:eqn:second_order_Helmholtz_splitting:magnetic}
	\end{align}
\end{subequations}
Here, the transversal fields are those that are $\scpro{\, \cdot \,}{\, \cdot \,}_{\eps}$- and $\scpro{\, \cdot \,}{\, \cdot \,}_{\mu}$-orthogonal to the gradient fields. For a rigorous definition of the gradient $\nabla$, the curl $\nabla^{\times}$ and the divergence $\nabla \cdot$ we refer to \cite[Appendix~A]{DeNittis_Lein:adiabatic_periodic_Maxwell_PsiDO:2013}. 
% subsection second_order_formalism_wave_equations_for_electric_and_magnetic_fields (end)

\subsection{Equivalence of first- and second-order equations} % (fold)
\label{setup:equivalence}
Put as a mathematical statement, the equivalence reads: 
\begin{theorem}[Equivalence of first- and second-order formalism]\label{setup:thm:equivalence}
	Suppose the medium is described by material weights that satisfy Assumption~\ref{setup:assumption:material_weights}, and the (complex) initial state $\Phi = (\phi^E , \phi^H) \in \Hil$ satisfies the constraint equation~\eqref{setup:eqn:Maxwell_equations:constraint}. 
	\begin{enumerate}[(1)]
		\item Equations~\eqref{setup:eqn:Maxwell_equations}, \eqref{setup:eqn:wave_equation_electric} and \eqref{setup:eqn:wave_equation_magnetic} are all equivalent in the following sense: electric and magnetic part of the solution $\Psi(t) = \bigl ( \psi^E(t) , \psi^H(t) \bigr )$ to the first-order Maxwell equations~\eqref{setup:eqn:Maxwell_equations} satisfy the wave equations~\eqref{setup:eqn:wave_equation_electric} and \eqref{setup:eqn:wave_equation_magnetic}, respectively. Conversely, we can reconstruct \emph{the} solution to Maxwell's equations~\eqref{setup:eqn:Maxwell_equations} from solution to the electric or magnetic wave equation alone. And hence, we can obtain \emph{the} solution to the magnetic wave equation~\eqref{setup:eqn:wave_equation_magnetic} from that of the electric wave equation~\eqref{setup:eqn:wave_equation_electric} and vice versa. 
		%
		% \item Then first- and second-order equations are equivalent, \ie electric and magnetic part of the solution $\Psi(t) = \bigl ( \psi^E(t) , \psi^H(t) \bigr )$ to the first-order Maxwell equations~\eqref{setup:eqn:Maxwell_equations} satisfy \eqref{setup:eqn:wave_equation_electric} and \eqref{setup:eqn:wave_equation_magnetic}, respectively. Conversely, we can reconstruct the solution to Maxwell's equations~\eqref{setup:eqn:Maxwell_equations} from solution to the electric or magnetic wave equation alone. 
		%
		\item The real-valued electromagnetic field $\bigl ( \mathbf{E}(t) , \mathbf{H}(t) \bigr )$ is represented by the complex solution $\Psi(t) = Q \bigl ( \mathbf{E}(t) , \mathbf{H}(t) \bigr )$ with the help of the injective map $Q : L^2(\R^3,\R^6) \longrightarrow \Hil$ defined through equation~\eqref{setup:eqn:Q_complex_representative_map}. 
	\end{enumerate}
\end{theorem}
When the material weights $W = \overline{W}$ are \emph{real}, this has long been known (see \eg \cite{Wilcox:scattering_theory_classical_physics:1966,Reed_Simon:scattering_theory_wave_equations:1977}). However, for \emph{complex} weights $W \neq \overline{W}$ this is new. Interestingly, the hurdle for an extension to complex material weights was to find \emph{physically meaningful first-order} Maxwell equations~\eqref{setup:eqn:Maxwell_equations} that are compatible with the real-valuedness of the physical fields $(\mathbf{E},\mathbf{H})$. Their derivation was one of the main aims of a recent work of ours \cite[Section~2]{DeNittis_Lein:Schroedinger_formalism_classical_waves:2017}, and this work is as an addendum. 

We emphasize that Theorem~\ref{setup:thm:equivalence} applies to media with real material weights as well: while we could equivalently work with real electric fields $\mathbf{E}$ \emph{directly} in equations~\eqref{setup:eqn:wave_equation_electric} and \eqref{setup:eqn:Maxwell_equations}, the alternate strategy to represent $\mathbf{E} = 2 \Re \psi^E$ as a complex $\omega \geq 0$ wave works just as well and has the added advantage of extending to complex $\eps$ and $\mu$. 

The proof consists of two parts, and only the first step requires a bit of work: firstly, we need to establish that our definition of $M_{EE}^2$ and $M_{HH}^2$ via the \emph{auxiliary} Maxwell operator (that lacks any frequency restriction) is compatible with the frequency constraint. The second step consists of showing a connection between the second initial condition $\partial_t \psi^E(t_0)$ for the second-order equation and $\psi^H(t_0)$ for the electric field wave equation~\eqref{setup:eqn:wave_equation_electric}, and an analogous statement for \eqref{setup:eqn:wave_equation_magnetic}.

\subsubsection{Equations~\eqref{setup:eqn:wave_equation_electric} and \eqref{setup:eqn:wave_equation_magnetic}, and the $\omega \geq 0$ frequency constraint} % (fold)
\label{setup:equivalence:spaces}
The main ingredient in the proof of the first step is the map
\begin{subequations}\label{setup:eqn:inclusion}
	\begin{align}
		\imath^E &: \ker \, ( \nabla \cdot \eps ) \longrightarrow \mathcal{J}_+
		, 
		\qquad 
		\psi^E \mapsto \left (
		\begin{matrix}
			\psi^E \\
			- \ii \mu^{-1} \, \nabla^{\times} \, \bigl ( M_{EE}^2 \bigr )^{- \nicefrac{1}{2}} \psi^E \\
		\end{matrix}
		\right )
		, 
		\label{setup:eqn:inclusion:electric}
		\\
		\imath^H &: \ker \, ( \nabla \cdot \mu ) \longrightarrow \mathcal{J}_+
		, 
		\qquad 
		\psi^H \mapsto \left (
		\begin{matrix}
			+ \ii \eps^{-1} \, \nabla^{\times} \bigl ( M_{HH}^2 \bigr )^{- \nicefrac{1}{2}} \psi^H \\
			\psi^H \\
		\end{matrix}
		\right )
		, 
		\label{setup:eqn:inclusion:magnetic}
	\end{align}
\end{subequations}
for the electric field and for the magnetic field. Figotin and Klein gave the explicit map to reconstruct the electric part from the magnetic part as \cite[equations~(10)–(11)]{Figotin_Klein:localization_classical_waves_II:1997}, although they assumed that $\eps$ and $\mu$ are real and scalar-valued; the difference in sign in \cite[equations~(10)]{Figotin_Klein:localization_classical_waves_II:1997} stems from their choice to use $(\mathbf{H},\mathbf{E})$ rather than $(\mathbf{E},\mathbf{H})$ as the electromagnetic field (\cf \cite[equation~(4)]{Figotin_Klein:localization_classical_waves_II:1997}). 

For the sake of concreteness, let us focus on the electric field. It turns out that $\imath^E$ is a bounded injection that maps transversal electric fields onto transversal, positive frequency \emph{electromagnetic} fields; its left-inverse $\mathrm{pr}^E : \Psi = \left ( 
\begin{smallmatrix}
	\psi^E \\
	\psi^H \\
\end{smallmatrix}
\right ) \mapsto \psi^E$ discards the magnetic component. 
\begin{proposition}\label{setup:prop:identification_spaces}
	Suppose Assumption~\ref{setup:assumption:material_weights} on the weights holds. 
	\begin{enumerate}[(1)]
		\item $\ker M_{EE}^2 = \ran \nabla = \ker M_{HH}^2$ where $\nabla$ is seen as an operator $\domain(\nabla) \subset L^2(\R^3,\C) \longrightarrow L^2(\R^3,\C^3)$. 
		\item The electric and magnetic parts of $\mathcal{J}_+$ coincide with the divergence-free electric and magnetic fields in the sense that 
		\begin{align*}
			L^2_{\eps}(\R^3,\C^3) \supset
			\ker \, ( \nabla \cdot \eps ) &= \mathcal{J}_+^E
			:= \Bigl \{ \psi^E \in L^2_{\eps}(\R^3,\C^3) \; \; \big \vert \; \; 
			\Psi = \left (
			\begin{smallmatrix}
				\psi^E \\
				\psi^H \\
			\end{smallmatrix}
			\right ) \in \mathcal{J}_+
			\Bigr \} 
			,
			\\
			L^2_{\mu}(\R^3,\C^3) \supset
			\ker \, ( \nabla \cdot \mu ) &= \mathcal{J}_+^H
			:= \Bigl \{ \psi^H \in L^2_{\mu}(\R^3,\C^3) \; \; \big \vert \; \; 
			\Psi = \left (
			\begin{smallmatrix}
				\psi^E \\
				\psi^H \\
			\end{smallmatrix}
			\right ) \in \mathcal{J}_+
			\Bigr \} 
			.
		\end{align*}
		\item The electric and magnetic components 
		\begin{align*}
			\Hil^E :& \negmedspace = \Bigl \{ \psi^E \in L^2_{\eps}(\R^3,\C^3) \; \; \big \vert \; \; 
			\Psi = \left (
			\begin{smallmatrix}
				\psi^E \\
				\psi^H \\
			\end{smallmatrix}
			\right ) \in \Hil
			\Bigr \} 
			= L^2_{\eps}(\R^3,\C^3)
			, 
			\\
			\Hil^H :& \negmedspace = \Bigl \{ \psi^H \in L^2_{\mu}(\R^3,\C^3) \; \; \big \vert \; \; 
			\Psi = \left (
			\begin{smallmatrix}
				\psi^E \\
				\psi^H \\
			\end{smallmatrix}
			\right ) \in \Hil
			\Bigr \} 
			= L^2_{\mu}(\R^3,\C^3)
			, 
		\end{align*}
		of the non-negative frequency space $\Hil$ coincide with $L^2_{\eps}(\R^3,\C^3)$ and $L^2_{\mu}(\R^3,\C^3)$. 
	\end{enumerate}
\end{proposition}
This Proposition tells us several things: on a conceptual level it states that the subspace of transversal electric fields $\ker \, ( \nabla \cdot \eps )$ does not contain (unphysical) waves that cannot be mapped to a real, transversal electromagnetic field $(\mathbf{E},\mathbf{H}) \in L^2(\R^3,\R^6)$. Indeed, concatenating $Q$ with the projection $\mathrm{pr}^E$ gives us an \emph{injective} map $\mathrm{pr}^E \circ Q : L^2(\R^3,\R^6) \longrightarrow L^2_{\eps}(\R^3,\C^3)$. 

And on a practical level it tells us that we could have defined $M_{EE}^2$ and $M_{HH}^2$ via $M^2$ rather than $(\Maux)^2$. Note, however, that $M^2 \neq M_{EE}^2 \oplus M_{HH}^2$ is not the direct sum of two operators because the Hilbert space $M$ is defined on does not — the $\omega \geq 0$ condition imposes a relation on electric and magnetic components. 
\begin{proof}
	Since the roles of electric and magnetic fields are symmetric, it suffices to make the arguments explicit only for the electric field. 
	\begin{enumerate}[(1)]
		\item We note that $M_{EE}^2 \psi^E = \eps^{-1} \, \nabla^{\times} \, \mu^{-1} \, \nabla^{\times} \psi^E = 0$ implies that either $\nabla^{\times} \psi^E = 0$ or $\mu^{-1} \, \nabla^{\times} \psi^E \in \ran \nabla = \ker \nabla^{\times}$ (as $\eps^{-1}$ and $\mu^{-1}$ are bounded and have bounded inverses). Thanks to $\nabla \times \nabla f = 0$, we deduce the inclusion $\ran \nabla \subseteq \ker M_{EE}^2$. 
		
		In the second case we are looking for a vector that lies in the intersection 
		\begin{align*}
			\ran \bigl ( \mu^{-1} \, \nabla^{\times} \bigr ) \cap \ran \nabla = \{ 0 \}
			, 
		\end{align*}
		and we will show it is necessarily zero: the adapted Helmholtz decomposition \cite[Appendix~A]{DeNittis_Lein:adiabatic_periodic_Maxwell_PsiDO:2013} of $L^2_{\mu}(\R^3,\C^3) = \ran \nabla \oplus \bigl ( \ran \nabla \bigr )^{\perp_{\mu}}$ implies 
		\begin{align*}
			\ran \bigl ( \mu^{-1} \, \nabla^{\times} \bigr ) &= \ker \, ( \nabla \cdot \mu ) 
			= \bigl ( \ran \nabla \bigr )^{\perp_{\mu}} 
			\subset L^2_{\mu}(\R^3,\C^3)
			. 
		\end{align*}
		Because $L^2_{\mu}(\R^3,\C^3) = L^2(\R^3,\C^3) = L^2_{\eps}(\R^3,\C^3)$ all coincide as \emph{Banach} spaces by our assumptions on $\eps$ and $\mu$, we conclude $\ran \bigl ( \mu^{-1} \, \nabla^{\times} \bigr ) \cap \ran \nabla = \{ 0 \}$ independently of the choice of scalar product. This shows the opposite inclusion $\ran \nabla \supseteq \ker M_{EE}^2$, and in combination yields $\ran \nabla = \ker M_{EE}^2$. 
		\item By the Helmholtz composition $L^2_{\eps}(\R^3,\C^3) = \ran \nabla \oplus \ker \, ( \nabla \cdot \eps ) = \ker M_{EE}^2 \oplus \ker \, ( \nabla \cdot \eps )$ holds, and since $M_{EE}^2 \geq 0$ is non-negative, we can express the second summand as 
		\begin{align*}
			\ker \, ( \nabla \cdot \eps ) = 1_{(0,\infty)} \bigl ( M_{EE}^2 \bigr ) \bigl [ L^2_{\eps}(\R^3,\C^3) \bigr ] 
			. 
		\end{align*}
		The same reasoning applies to $\mathcal{J}_+ = 1_{(0,\infty)}(\Maux) \bigl [ L^2_W(\R^3,\C^6) \bigr ]$: the gradient fields $\mathcal{G} = 1_{\{ 0 \}}(\Maux) \bigl [ L^2_W(\R^3,\C^6) \bigr ]$ (this time seen as a subspace of $L^2_W(\R^3,\C^6)$) make up the kernel of $\Maux$ and spectral calculus gives us for free that $\mathcal{J}_+$ is $\scpro{\, \cdot \,}{\, \cdot \,}_W$-orthogonal to the gradient fields. Manually writing out the orthogonality condition then yields that elements of $\mathcal{J}_+ \subset \ker \, (\Div \, W) = \mathcal{G}^{\perp_W} \subset L^2(\R^3,\C^6)$ satisfy the divergence-free condition $\Div \, W \, \Psi = 0$, which we can write out as $\nabla \cdot \eps \, \psi^E = 0 = \nabla \cdot \mu \, \psi^H$. This shows the inclusion $\mathcal{J}_+^E \subseteq \ker \, ( \nabla \cdot \eps )$. 
		
		The other inclusion $\mathcal{J}_+^E \supseteq \ker \, ( \nabla \cdot \eps )$ requires a bit more work. The key here is the map $\imath^E$ defined by \eqref{setup:eqn:inclusion:electric} that associates to each divergence-free electric field an element in $\mathcal{J}_+$. Of course, we need to show that this map is well-defined as a map $\ker \, ( \nabla \cdot \eps ) \longrightarrow L^2_W(\R^3,\C^6)$ and then prove that it maps divergence-free electric fields onto \emph{positive} frequency electromagnetic fields (which are automatically divergence-free). 
		
		To show that $\mu^{-1} \, \nabla^{\times} \, \bigl ( M_{EE}^2 \bigr )^{- \nicefrac{1}{2}} \psi^E$ defines an element of $L^2_{\mu}(\R^3,\C^3)$, we note that the squared auxiliary Maxwell operator 
		\begin{align*}
			(\Maux)^2 \big \vert_{\ker \, (\Div \, W)} = M_{EE}^2 \big \vert_{\ker \, ( \nabla \cdot \eps )} \oplus M_{HH}^2 \big \vert_{\ker \, ( \nabla \cdot \mu )}
		\end{align*}
		restricted to divergence-free fields (including negative frequency waves) coincides with the restriction of $M_{EE}^2$ and $M_{HH}^2$ to the divergence-free electric and magnetic fields. Then for all divergence-free electromagnetic fields $\Psi = (\psi^E , \psi^H) \in \ker \, (\Div \, W)$ the norm of the operator 
		\begin{align}
			\mathrm{sgn} \bigl ( \Maux \bigr ) \, \Psi &= \Maux \, \bigl ( (\Maux)^2 \bigr )^{-\nicefrac{1}{2}} \Psi 
			= \left (
			\begin{matrix}
				+ \ii \eps^{-1} \, \nabla^{\times} \, \bigl ( M_{HH}^2 \bigr )^{- \nicefrac{1}{2}} \psi^H \\
				- \ii \mu^{-1} \, \nabla^{\times} \, \bigl ( M_{EE}^2 \bigr )^{- \nicefrac{1}{2}} \psi^E \\
			\end{matrix}
			\right )
			\label{setup:eqn:sgn_M_EE_2_equation}
			\\
			&= \bigl ( (\Maux)^2 \bigr )^{-\nicefrac{1}{2}} \, \Maux \Psi 
			= \left (
			\begin{matrix}
				+ \bigl ( M_{EE}^2 \bigr )^{- \nicefrac{1}{2}} \, \ii \eps^{-1} \, \nabla^{\times} \, \psi^H \\
				- \bigl ( M_{HH}^2 \bigr )^{- \nicefrac{1}{2}} \, \ii \mu^{-1} \, \nabla^{\times} \psi^E \\
			\end{matrix}
			\right )
			\label{setup:eqn:sgn_M_EE_2_equation_alternate} 
		\end{align}
		is necessarily bounded by $1 \cdot \snorm{\Psi}_W$: to justify \eqref{setup:eqn:sgn_M_EE_2_equation} we note that $\abs{\Maux} = \bigl ( (\Maux)^2 \bigr )^{\nicefrac{1}{2}}$ maps $\domain(\Maux) \cap \ker \, \bigl ( \Div \, W \bigr )$ onto $\ran \Maux = \ker \, \bigl ( \Div \, W \bigr )$, so that its inverse 
		\begin{align*}
			\bigl ( (\Maux)^2 \bigr )^{- \nicefrac{1}{2}} = \babs{\Maux}^{-1} : \ker \, \bigl ( \Div \, W \bigr ) \longrightarrow \domain(\Maux) \cap \ker \, \bigl ( \Div \, W \bigr )
		\end{align*}
		maps divergence-free fields onto divergence-free fields from the domain. Hence, applying $\Maux$ from the left gives us a \emph{bounded} operator on $\ker \, \bigl ( \Div \, W \bigr )$. 
		
		Therefore, $\nabla^{\times} \, \bigl ( M_{EE}^2 \bigr )^{- \nicefrac{1}{2}} : \ker \, ( \nabla \cdot \eps ) \longrightarrow \ran \nabla^{\times} = \ker \, ( \nabla \cdot )$ maps to the divergence-free fields, so if we multiply from the left with $\mu^{-1}$, we indeed get something in $\ker \, (\nabla \cdot \mu)$. Moreover, we deduce that this is in fact a \emph{bounded} operator, and $\imath^E$ is well-defined as a map 
		\begin{align*}
			\ker \, ( \nabla \cdot \eps ) \longrightarrow \ker \, (\Div \, W) = \ker \, ( \nabla \cdot \eps ) \oplus \ker \, ( \nabla \cdot \mu ) 
			. 
		\end{align*}
		All that is left is to prove that $\imath^E(\psi^E)$ is a positive frequency wave. The reason we defined $\imath^E$ the way we did is readily apparent when we compare it with \eqref{setup:eqn:sgn_M_EE_2_equation}: the operator $\mathrm{sgn} \bigl ( \Maux \bigr ) \big \vert_{\ker \, (\Div \, W)}$ has two eigenvalues, $\pm 1$, and the positive frequency fields $\Psi \in \mathcal{J}_+$ form the eigenspace of $\mathrm{sgn} \bigl ( \Maux \big \vert_{\ker \, (\Div \, W)} \bigr )$ to the eigenvalue $+1$. By design $\imath^E$ gives us an eigenvector to the eigenvalue $+1$: because we can write $\mathrm{sgn}(\Maux)$ in both ways, \eqref{setup:eqn:sgn_M_EE_2_equation} and \eqref{setup:eqn:sgn_M_EE_2_equation_alternate}, we deduce 
		\begin{align*}
			\ii \eps^{-1} \, \nabla^{\times} \, \bigl ( M_{HH}^2 \bigr )^{- \nicefrac{1}{2}} \psi^H = \bigl ( M_{EE}^2 \bigr )^{- \nicefrac{1}{2}} \, \ii \eps^{-1} \, \nabla^{\times} \, \psi^H
		\end{align*}
		and a similar equation for the other component. The important fact here is that $M_{HH}^2$ becomes $M_{EE}^2$ when we “commute” it with $+ \ii \eps^{-1} \, \nabla^{\times}$, which is what we will exploit when applying $\mathrm{sgn}(\Maux)$ to $\imath^E(\psi^E)$, 
		\begin{align*}
			\mathrm{sgn}(\Maux) \, \imath^E(\psi^E) &= \left (
			\begin{matrix}
				\eps^{-1} \, \nabla^{\times} \, \bigl ( M_{HH}^2 \bigr )^{- \nicefrac{1}{2}} \mu^{-1} \, \nabla^{\times} \, \bigl ( M_{EE}^2 \bigr )^{- \nicefrac{1}{2}} \, \psi^E \\
				- \ii \mu^{-1} \, \nabla^{\times} \, \bigl ( M_{EE}^2 \bigr )^{- \nicefrac{1}{2}} \psi^E \\
			\end{matrix}
			\right )
			\\
			&= (+1) \, \left (
			\begin{matrix}
				\eps^{-1} \, \nabla^{\times} \, \mu^{-1} \, \nabla^{\times} \, \bigl ( M_{EE}^2 \bigr )^{-1} \psi^E \\
				- \ii \mu^{-1} \, \nabla^{\times} \, \bigl ( M_{EE}^2 \bigr )^{- \nicefrac{1}{2}} \psi^E \\
			\end{matrix}
			\right )
			= + \imath^E(\psi^E)
			.
		\end{align*}
		Note that we may either use $\psi^E$ or $\psi^H$ as an independent variable, and choosing $\psi^H$ instead of $\psi^E$ gives us $\imath^H$.
		
		Evidently, its inverse — the projection $\mathrm{pr}^E : (\psi^E , \psi^H) \mapsto \psi^E$ onto the first component is bounded and a left-inverse to $\imath^E$, \ie we have $\mathrm{pr}^E \circ \imath^E = \id_{\ker \, ( \nabla \cdot \eps )}$. This shows the opposite inclusion, $\mathcal{J}_+^E \supseteq \ker \, ( \nabla \cdot \eps )$. 
		\item This follows from the Helmholtz decompositions~\eqref{setup:eqn:second_order_Helmholtz_splitting} combined with (1) and (2). 
	\end{enumerate}
\end{proof}
\begin{remark}
	A second way to verify that $\imath^E(\psi^E)$ is composed solely of positive frequencies relies on functional calculus for $M_{EE}^2 = \eps^{-1} \, \nabla^{\times} \, \mu^{-1} \nabla^{\times}$. If we assume for a moment that $\imath^E(\psi^E) \in \mathcal{D}(\Maux)$, then the following expressions are all well-defined: 
	\begin{align*}
		\Maux \, \imath^E(\psi^E) 
		&= \int_{(0,\infty)} \Maux \, \left (
		\begin{matrix}
			\dd 1_{\lambda} \bigl ( M_{EE}^2 \bigr ) \, \psi^E \\
			- \ii \lambda^{-\nicefrac{1}{2}} \, \mu^{-1} \, \nabla^{\times} \, \dd 1_{\lambda} \bigl ( M_{EE}^2 \bigr ) \, \psi^E \\
		\end{matrix}
		\right )
		\\
		&
		= \int_{(0,\infty)} \left (
		\begin{matrix}
			\lambda^{-\nicefrac{1}{2}} \, \eps^{-1} \, \nabla^{\times} \, \mu^{-1} \nabla^{\times} \, \dd 1_{\lambda} \bigl ( M_{EE}^2 \bigr ) \, \psi^E \\
			- \ii \, \mu^{-1} \, \nabla^{\times} \, \dd 1_{\lambda} \bigl ( M_{EE}^2 \bigr ) \, \psi^E \\
		\end{matrix}
		\right )
		\\
		&= \int_{(0,\infty)} \bigl ( + \lambda^{\nicefrac{1}{2}} \bigr ) \, \left (
		\begin{matrix}
			\dd 1_{\lambda} \bigl ( M_{EE}^2 \bigr ) \, \psi^E \\
			- \ii \lambda^{-\nicefrac{1}{2}} \, \mu^{-1} \, \nabla^{\times} \, \dd 1_{\lambda} \bigl ( M_{EE}^2 \bigr ) \, \psi^E \\
		\end{matrix}
		\right )
	\end{align*}
	Of course, if $\imath^E(\psi^E) \not \in \mathcal{D}(\Maux)$, we need to regularize: if we replace $\imath^E(\psi^E)$ with the cut off wave $\Psi_b := \imath^E \Bigl ( 1_{(0,b)} \bigl ( M_{EE}^2 \bigr ) \, \psi^E \Bigr )$ for $b < \infty$, the above computation then shows that for $\Psi_b$ only frequencies in the range $(0,\sqrt{b})$ are excited, negative frequencies are excluded. The limit $\Maux \Psi_b$ as $b \rightarrow \infty$ will in general not exist in $L^2_W(\R^3,\C^6)$, but nevertheless, we still see that $\imath^E(\psi^E)$ is composed solely of positive frequencies.
\end{remark}
When we showed part (2), we have indeed also furnished a proof for the following Corollary that will be useful for our discussion of periodic electromagnetic media in Section~\ref{periodic}: 
\begin{corollary}\label{setup:cor:inclusion_bounded_injections}
	The maps $\imath^E : \ker \, ( \nabla \cdot \eps ) \longrightarrow \mathcal{J}_+$ and $\imath^H : \ker \, ( \nabla \cdot \mu ) \longrightarrow \mathcal{J}_+$ from equations~\eqref{setup:eqn:inclusion} are bounded injections with left-inverses $\mathrm{pr}^{E,H} : (\psi^E , \psi^H) \mapsto \psi^{E,H}$.
\end{corollary}
Another way the equivalence of electric and magnetic field equations manifests itself is as a \emph{unitary} equivalence of the positive frequency parts of $M_{EE}^2$ and $M_{HH}^2$. 
\begin{corollary}
	The positive frequency part $M_{HH}^2 \vert_{\omega > 0} := M_{HH}^2 \, \big \vert_{\ker \, ( \nabla \cdot \mu )}$ is related to the positive frequency part of 
	\begin{align*}
		M_{EE}^2 \vert_{\omega > 0} = U_{EH} \, M_{HH}^2 \vert_{\omega > 0} \, U_{EH}^{-1} 
	\end{align*}
	via the \emph{unitary} map 
	\begin{align*}
		U_{EH} := \mathrm{pr}^E \circ \imath^H 
		% = + \ii \eps^{-1} \, \nabla^{\times} \, \bigl ( M_{HH}^2 \bigr )^{-\nicefrac{1}{2}}
		: \ker \, ( \nabla \cdot \mu ) \subset L^2_{\mu}(\R^3,\C^3) \longrightarrow \ker \, ( \nabla \cdot \eps ) \subset L^2_{\eps}(\R^3,\C^3)
		. 
	\end{align*}
\end{corollary}
This was already recognized by Figotin and Klein (\cf \cite[equations~(9)–(10)]{Figotin_Klein:localization_classical_waves_II:1997}), although they did not give a proof showing that $U_{EH}$ is well-defined. 
\begin{proof}
	The well-definedness and invertibility of $U_{EH}$ are direct consequences of Proposition~\ref{setup:prop:identification_spaces}. Therefore unitarity follows from checking 
	\begin{align*}
		\bscpro{U_{EH} \varphi^H}{U_{EH} \psi^H}_{\eps} &= \Bscpro{\mu^{-1} \, \nabla^{\times} \, \eps^{-1} \, \nabla^{\times} \, (M_{HH}^2)^{-\nicefrac{1}{2}} \varphi^H}{(M_{HH}^2)^{-\nicefrac{1}{2}} \psi^H}_{\mu}
		\\
		&= \Bscpro{\varphi^H}{(M_{HH}^2)^{+\nicefrac{1}{2}} \, (M_{HH}^2)^{-\nicefrac{1}{2}} \psi^H}_{\mu} 
		= \scpro{\varphi^H}{\psi^H}_{\mu} 
	\end{align*}
	by direct computation on the dense subset $\domain(M_{HH}^2) \cap \ker \, ( \nabla \cdot \mu )$, and extending this by density to all of $\ker \, ( \nabla \cdot \mu )$. 
\end{proof}
%  
% subsection equations_eqref_setup_eqn_wave_equations (end)

\subsubsection{Equivalence of the dynamics} % (fold)
\label{setup:equivalence:dynamics}
The actual proof of equivalence is completely standard, indeed the hard part was in properly defining Maxwell's equations and verify that the spaces on which the second-order equations are defined are correct. 
\begin{proof}[Theorem~\ref{setup:thm:equivalence}]
	As before, we will only formulate the proof for the electric field. Moreover, to simplify the presentation we impose in addition that the complex initial condition $\Phi := Q (\mathbf{E}_0 , \mathbf{H}_0) \in \mathcal{D}(M^2) \cap \mathcal{J}_+$ holds. This just allows us to write out the proof for strong rather than weak solutions. 
	
	Thanks to Theorem~\ref{setup:thm:equivalence_frameworks} the (strong) solution $\Psi(t) = \e^{- \ii (t - t_0) M} \Phi$ to Maxwell's equations~\eqref{setup:eqn:Maxwell_equations} can be expressed in terms of the Maxwell operator. Seeing as $\e^{- \ii t M} : \domain(M^2) \longrightarrow \domain(M^2)$ preserves the core $\domain(M^2)$, the second-order time-derivative of $\Psi(t)$ exists in $\Hil$. Separating out the electric components, $\partial_t^2 \Psi(t) + M^2 \Psi(t) = 0$ yields~\eqref{setup:eqn:wave_equation_electric:dynamics}. By definition of $\Hil$ and the transversality of the initial condition, also the second initial condition $\partial_t \psi^E(t_0) = - \mu^{-1} \, \nabla \times \psi^H(t_0)$ of the wave equation is satisfied. Hence, the magnetic and electric part satisfy \eqref{setup:eqn:wave_equation_electric:dynamics}. Lastly, the transversality condition~\eqref{setup:eqn:wave_equation_electric:constraint} is preserved (Proposition~\ref{setup:prop:complex_representation_real_fields}~(2)). This shows that such a strong solution to Maxwell's equations yields a strong solution to the wave equation~\eqref{setup:eqn:wave_equation_electric} for the electric field. 
	
	Now conversely, suppose $\psi^E(t)$ solves \eqref{setup:eqn:wave_equation_electric}. Our additional assumption $\Phi \in \mathcal{D}(M^2)$ guarantees that it is indeed a strong solution as the first- and second-order time-derivatives exist in $L^2_{\eps}(\R^3,\C^3)$. Writing \eqref{setup:eqn:wave_equation_electric} as a first-order equation yields 
	\begin{align*}
		\frac{\partial}{\partial t} \left (
		\begin{matrix}
			\psi^E(t) \\
			\eta^E(t) \\
		\end{matrix}
		\right ) = \left (
		\begin{matrix}
			0 & \id \\
			M_{EE}^2 & 0 \\
		\end{matrix}
		\right ) \left (
		\begin{matrix}
			\psi^E(t) \\
			\eta^E(t) \\
		\end{matrix}
		\right )
		. 
	\end{align*}
	However, instead of using $\eta^E(t) = \partial_t \psi^E(t)$ as the second variable, we can introduce the magnetic field $\psi^H(t) := \bigl ( \nabla^{\times} \bigr )^{-1} \, \eps \, \eta^H(t)$. This change of variables makes sense as $\eta^E(t) \in \ker \, ( \nabla \cdot \eps )$ and the curl has the inverse $\bigl ( \nabla^{\times} \bigr )^{-1} : \ker (\nabla \cdot \,) \longrightarrow L^2_{\eps}(\R^3,\C^3)$ for divergence-free fields. Using the complex magnetic field as the second variable, we obtain the equations
	\begin{align*}
		\left (
		\begin{matrix}
			\id & 0 \\
			0 & \eps^{-1} \, \nabla^{\times} \\
		\end{matrix}
		\right ) \; \frac{\partial}{\partial t} \left (
		\begin{matrix}
			\psi^E(t) \\
			\psi^H(t) \\
		\end{matrix}
		\right )
		= \left (
		\begin{matrix}
			0 & \id \\
			M_{EE}^2 & 0 \\
		\end{matrix}
		\right ) 
		\left (
		\begin{matrix}
			\id & 0 \\
			0 & \eps^{-1} \, \nabla^{\times} \\
		\end{matrix}
		\right ) \left (
		\begin{matrix}
			\psi^E(t) \\
			\psi^H(t) \\
		\end{matrix}
		\right )
		, 
	\end{align*}
	which can be rewritten as the dynamical Maxwell equation~\eqref{setup:eqn:Maxwell_equations:dynamics}. The latter step is once again allowed because $\bigl ( \nabla^{\times} \bigr )^{-1} \, \eps : \ker \, ( \nabla \cdot \eps ) \longrightarrow L^2_{\eps}(\R^3,\C^3)$ is bounded. Moreover, with our specific choice of initial condition $\Phi = Q (\mathbf{E}_0,\mathbf{H}_0) \in \Hil$, the resulting electromagnetic field $\Psi(t) = \bigl ( \psi^E(t) , \psi^H(t) \bigr )$ indeed is a positive frequency wave. 
	
	This is the proof of the statement for strong solutions. But of course, the extra assumption $\Phi \in \domain(M^2)$ on the initial condition can be dropped if we work with weak solutions. In essence, we impose the extra condition on the test functions $\Theta \in \domain(M^2)$ and exploit that $\domain(M^2)$ lies densely in $\Hil$. 
\end{proof}
\begin{remark}[Extension to other dimensions]\label{setup:rem:other_dimensions}
	Evidently, none of our arguments rely on the fact that the spatial domain is \emph{all} of $\R^3$ rather than some subset. We may want to work on a subdomain of the form $\R^2 \times [0,h] \subset \R^3$ to model a quasi-2d waveguide slab, for example. Here, a proper choice of boundary conditions such as those for a perfect electric or magnetic conductor on the upper and lower plate are necessary to define Maxwell's equations — and, by extension, the Maxwell operators. 
	
	Consequently, the equivalence of Chern numbers applies to two- and three-dimensional topological photonic crystals alike. 
\end{remark}
%
% subsection equivalence_of_the_dynamics (end)
% section equivalence_of_first_and_second_order_equations (end)
% section  (end)

%% file: section_3.tex
%!TEX root = /Users/max/Dropbox/research/photonic crystals/active/equivalence first- and second-order formalism/paper/equivalence first- and second-order formalism electromagnetism.tex
\section{Equivalence of frequency band pictures in periodic media} % (fold)
\label{periodic}
The previously proven equivalence of first- and second-order dynamics, Theorem~\ref{setup:thm:equivalence_frameworks}, evidently applies to the special case of periodic electromagnetic media, better known as \emph{photonic crystals.} 
\begin{assumption}[Periodic weights]\label{periodic:assumption:periodic_material_weights}
	Suppose the material weights satisfy Assumption~\ref{setup:assumption:material_weights} and there exists a lattice $\Gamma \cong \Z^3$ so that $W(x + \gamma) = W(x)$ holds for all $\gamma \in \Gamma$ and almost all $x \in \R^3$
\end{assumption}
\begin{remark}[Other dimensions]
	As we have mentioned in Remark~\ref{setup:rem:other_dimensions} none of our arguments below are specific to three-di\-men\-sional photonic crystals, and analogs of our main results, Proposition~\ref{periodic:prop:reconstruction_eigenfunctions} and Theorem~\ref{periodic:prop:equivalence_vector_bundles}, also hold true for photonic crystals of other dimensions. 
\end{remark}
Bloch-Floquet theory — just like in case of periodic Schrödinger operators \cite{Kuchment:Floquet_theory:1993,Kuchment:math_photonic_crystals:2001,DeNittis_Lein:adiabatic_periodic_Maxwell_PsiDO:2013} — gives rise to \emph{three} \emph{sets} of frequency bands for the periodic operators $M$, $M_{EE}^2$ and $M_{HH}^2$. The purpose of this section is to show that these three sets of frequency bands (Proposition~\ref{periodic:prop:reconstruction_eigenfunctions}) and their frequency band \emph{topologies} (as measured by Chern numbers) coincide (Theorem~\ref{intro:thm:Chern_numbers_agree} and Proposition~\ref{periodic:prop:equivalence_vector_bundles}). The key ingredients are the maps $\imath^{E,H}$ from equation~\eqref{setup:eqn:inclusion} and $\mathrm{pr}^{E,H} : (\psi^E , \psi^H) \mapsto \psi^{E,H}$, that allow us to relate electromagnetic, electric and magnetic Bloch functions as well as three associated vector bundles with one another.

\subsection{The frequency band spectra coincide} % (fold)
\label{periodic:frequency_band_spectra}
Showing the equivalence of the frequency band spectra is different from proving the equivalence of the dynamical equations, and not merely a corollary of Theorem~\ref{setup:thm:equivalence}: to \emph{uniquely} fix a solution to the second-order equations we need the electric \emph{and} the magnetic fields as an \emph{input}; in contrast, we can reconstruct electromagnetic eigenfunctions solely from the electric or magnetic components \emph{alone.} 

Let us start by providing some of the basics on periodic operators.

\subsubsection{Exploiting periodicity: the Bloch-Floquet-Zak representation} % (fold)
\label{periodic:frequency_band_spectra:Zak}
This subsection collects basic facts about a variant of the discrete Fourier transform
\begin{align}
	(\Fourier \Psi)(k,x) = \sum_{\gamma \in \Gamma} \e^{- \ii k \cdot (x + \gamma)} \, \Psi(x + \gamma)
	\label{periodic:eqn:Bloch_Floquet_Zak_transform}
\end{align}
that is commonly called the Zak transform \cite{Zak:dynamics_Bloch_electrons:1968}; compared to the more common Bloch-Floquet transform it includes the extra phase factor $\e^{- \ii k \cdot x}$ in its definition. Experts on the subject may proceed directly to Section~\ref{periodic:frequency_band_spectra:fiber_decomposition}. 

Periodicity with respect to a lattice $\Gamma$ in these systems is exploited by decomposing position $q  = \gamma + x \in \R^3 \cong \Gamma \times \WS$ into a lattice coordinate $\gamma \in \Gamma$ and a position $x \in \WS$ located in a fundamental cell, usually referred to as the Wigner-Seitz cell $\WS$. Similarly, momenta $p = k + \gamma^* \in \R^3 \cong \BZ \times \Gamma^*$ are expressed as the sum of Bloch momentum $k \in \BZ$ that is taken from the first Brillouin zone $\BZ$ and a reciprocal lattice vector $\gamma^* \in \Gamma^*$; here, the dual lattice 
\begin{align*}
	\Gamma^* := \mathrm{span}_{\Z} \, \bigl \{ e_1^* , e_2^* , e_3^* \bigr \} 
\end{align*}
can be constructed from the real space lattice $\Gamma = \mathrm{span}_{\Z} \, \bigl \{ e_1 , e_2 , e_3 \bigr \}$ by requiring its basis vectors $e_n^*$ satisfy $e_j \cdot e_n^* = 2 \pi \, \delta_{jn}$ \cite{Grosso_Parravicini:solid_state_physics:2003}. On the level of groups, we may view the Brillouin zone $\BZ \simeq \widehat{\Gamma} \simeq \T^3$ and $\WS \simeq \widehat{\Gamma^*} \simeq \T^3$ as the \emph{dual groups} to $\Gamma \simeq \Z^3$ and $\Gamma^* \simeq \Z^3$, which is why we will identify the real space unit cell $\WS \simeq \T^3$ with a torus. 

Ordinarily, the Zak transform is defined as a unitary map 
\begin{align*}
	\Fourier : L^2(\R^3,\C^n) \longrightarrow L^2_{\mathrm{eq}} \bigl ( \R^3 \, , \, L^2(\T^3,\C^n) \bigr ) \cong L^2(\BZ) \otimes L^2(\T^3,\C^n) 
\end{align*}
where in the last step we have canonically identified the space of equivariant $L^2$-functions 
\begin{align*}
	L^2_{\mathrm{eq}} &\bigl ( \R^3 \, , \, L^2(\T^3,\C^n) \bigr ) := 
	\\
	&:= \Bigl \{ \Psi \in L^2_{\mathrm{loc}} \bigl ( \R^3 \, , \, L^2(\T^3,\C^n) \bigr ) \; \; \big \vert \; \; \Psi(k - \gamma^*,x) = \e^{+ \ii \gamma^* \cdot x} \, \Psi(k,x) \mbox{ a.~e.~$\forall \gamma^* \in \Gamma^*$} \Bigr \}
\end{align*}
with a tensor product space by restricting equivariant $L^2$-functions to the unit cell that contains $k = 0$. That is because Zak transformed functions are $\Gamma$-periodic in $x$ and $\Gamma^*$-quasiperiodic in $k$, 
\begin{subequations}\label{periodic:eqn:Zak_periodicity_relations}
	\begin{align}
		(\Fourier \Psi)(k , x - \gamma) &= (\Fourier \Psi)(k , x) 
		, 
		\label{periodic:eqn:Zak_periodicity_relations:position}
		\\
		(\Fourier \Psi)(k - \gamma^* , x) &= \e^{+ \ii \gamma^* \cdot x} \, (\Fourier \Psi)(k , x) 
		. 
		\label{periodic:eqn:Zak_periodicity_relations:momentum}
	\end{align}
\end{subequations}
These definitions extend naturally when the $L^2$-spaces are subjected to $\Gamma$-periodic weights, so that the Zak transform can then be considered as a unitary map
\begin{align*}
	\Fourier : L^2_W(\R^3,\C^6) \longrightarrow L^2(\BZ) \otimes L^2_W(\T^3,\C^6) 
\end{align*}
between \emph{weighted} electromagnetic $L^2$-spaces; here, $L^2_W(\T^3,\C^6)$ has been defined as the \emph{Banach space} $L^2(\T^3,\C^6)$ endowed with the energy scalar product 
\begin{align*}
	\sscpro{\phi}{\psi}_W := \bscpro{\phi}{W \, \psi}_{L^2(\T^3,\C^6)} 
\end{align*}
in analogy to our definition of $L^2_W(\R^3,\C^6)$ from Section~\ref{setup:first_order:representation_complex_waves}. In just the same way, we may view $\Fourier$ as a map on the electric or magnetic weighted $L^2$-spaces. 

The non-negative frequency space $\Fourier : \Hil \longrightarrow \int_{\BZ}^{\oplus} \dd k \, \Hil(k)$ splits into a \emph{direct integral of Hilbert spaces;} while the general definition (\cf  \cite[Part~II, Chapter~1, Section~5]{Dixmier:von_Neumann_algebras:1981}) is quite technical, especially when it comes to specifying what measurable means in this context, these subtleties are not of importance here: we can straight-forwardly express 
\begin{align*}
	\Hil(k) = \ran \, 1_{[0,\infty)} \bigl ( \Maux(k) \bigr ) \subset L^2_W(\T^3,\C^6) 
\end{align*}
in terms of spectral projections of $\Maux(k)$, whose $k$-dependence is even \emph{analytic} — and thus, measurable — on the set $\BZ \setminus \{ 0 \}$ that has full $\dd k$ measure. 
% subsubsection exploiting_periodicity_the_bloch_floquet_zak_representation (end)

\subsubsection{Fiber decomposition of the first- and second-order operators} % (fold)
\label{periodic:frequency_band_spectra:fiber_decomposition}
Periodic operators $A : \domain(A) \subseteq L^2(\R^3,\C^n) \longrightarrow L^2(\R^3,\C^n)$ are operators which commute with lattice translations. These therefore admit a fiber decomposition 
\begin{align*}
	\Fourier \, A \, \Fourier^{\, -1} &= \int_{\BZ}^{\oplus} \dd k \, A(k) 
	. 
\end{align*}
Due to the quasi periodicity condition~\eqref{periodic:eqn:Zak_periodicity_relations:momentum} for any $\gamma^* \in \Gamma^*$ the fiber operators $A(k)$ and 
\begin{align}
	A(k - \gamma^*) &= \e^{+ \ii \gamma^* \cdot \hat{x}} \, A(k) \, \e^{- \ii \gamma^* \cdot \hat{x}}
	\label{periodic:eqn:equivariant_operator}
\end{align}
are unitarily equivalent via the multiplication operator $\e^{- \ii \gamma^* \cdot \hat{x}}$, and we call operator-valued functions $k \mapsto A(k)$ that satisfy \eqref{periodic:eqn:equivariant_operator} \emph{equivariant}. 

Two particularly relevant examples are periodic multiplication operators such as the electromagnetic weights 
\begin{align*}
	\Fourier \, W \, \Fourier^{\, -1} &= \id_{L^2(\BZ)} \otimes W \equiv W 
\end{align*}
and the derivatives 
\begin{align*}
	\Fourier \, (- \ii \partial_j) \, \Fourier^{\, -1} &= \id_{L^2(\BZ)} \otimes (- \ii \partial_j) + \hat{k} \otimes \id_{L^2_W(\T^3,\C^6)}
\end{align*}
which are equipped with the obvious domains (\cf our discussion in \cite[Section~3.1]{DeNittis_Lein:adiabatic_periodic_Maxwell_PsiDO:2013}). 

Thus, the unitary $\Fourier$ facilitates a fiber decomposition of the Maxwell operators 
\begin{subequations}
	\begin{align}
		\Maux &\cong \int_{\BZ}^{\oplus} \dd k \, \Maux(k)
		= \int_{\BZ}^{\oplus} \dd k \, \left (
		\begin{matrix}
			0 & - \eps^{-1} \, (- \ii \nabla + k)^{\times} \\
			+ \mu^{-1} \, (- \ii \nabla + k)^{\times} & 0 \\
		\end{matrix}
		\right )
		, 
		\\
		M &\cong \int_{\BZ}^{\oplus} \dd k \, M(k) 
		= \int_{\BZ}^{\oplus} \dd k \, \Maux(k) \, \big \vert_{\omega \geq 0}
		, 
	\end{align}
\end{subequations}
where the fiber operators $\Maux(k)$ and $M(k)$ act on $L^2_W(\T^3,\C^6)$ and the non-negative frequency subspace $\Hil(k) \subset L^2_W(\T^3,\C^6)$, respectively. In the same way the wave operators 
\begin{subequations}
	\begin{align}
		M_{EE}^2 &\cong \int_{\BZ}^{\oplus} \dd k \, M_{EE}^2(k) 
		= \int_{\BZ}^{\oplus} \dd k \, \eps^{-1} \, (\nabla - \ii k)^{\times} \, \mu^{-1} \, (\nabla - \ii k)^{\times}
		, 
		\\
		M_{HH}^2 &\cong \int_{\BZ}^{\oplus} \dd k \, M_{HH}^2(k) 
		= \int_{\BZ}^{\oplus} \dd k \, \mu^{-1} \, (\nabla - \ii k)^{\times} \, \eps^{-1} \, (\nabla - \ii k)^{\times}
		, 
	\end{align}
\end{subequations}
split into direct integrals, and $M_{EE}^2(k)$ and $M_{HH}^2(k)$ are selfadjoint operators on $L^2_{\eps}(\T^3,\C^3)$ and $L^2_{\mu}(\T^3,\C^3)$, respectively. 

Spectral and analyticity properties of $\Maux(k)$ have been studied extensively in the past (\eg in \cite{Kuchment:math_photonic_crystals:2001,DeNittis_Lein:adiabatic_periodic_Maxwell_PsiDO:2013}). Apart from essential spectrum at $\omega_0(k) = 0$ due to gradient fields, $\sigma \bigl (\Maux(k) \bigr ) \setminus \{ 0 \} = \sigma_{\mathrm{disc}} \bigl (\Maux(k) \bigr ) \setminus \{ 0 \}$ is purely discrete; since $\Maux(k)$ is not bounded from below, the eigenvalues accumulate at $\pm \infty$ (\cf \cite[Theorem~1.4]{DeNittis_Lein:adiabatic_periodic_Maxwell_PsiDO:2013}). As $k$ varies, these eigenvalues form frequency bands $k \mapsto \omega_n(k)$ and both, $\omega_n(k)$ and the associated eigenfunctions can be chosen locally analytically away from band crossings. 

Due to the structure of the operator there are 2+2 “ground state bands” with approximately linear dispersion near $k = 0$ and $\omega = 0$ due to long-wavelength waves which to good approximation only see unit cell averages of the material weights $W$. 

Seeing as the domain $\domain \bigl ( \Maux(k) \bigr ) = \domain \bigl ( \Maux(0) \bigr )$ of the linear polynomial $k \mapsto \Maux(k)$ is independent of $k$, this operator is evidently analytic. Its non-negative frequency restriction $M(k) = \Maux(k) \, \vert_{\omega \geq 0}$, however, is more delicate, because of the singular behavior of the ground state Bloch functions at $k = 0$. 

Not surprisingly, the properties of the second-order operators $M_{EE}^2$ and $M_{HH}^2$ mirror those of $\Maux(k)$ and $M(k)$. 
\begin{lemma}\label{periodic:lem:essential_properties_wave_operators}
	Suppose the material weights satisfy Assumption~\ref{periodic:assumption:periodic_material_weights}. Then the following holds: 
	\begin{enumerate}[(1)]
		\item For all $k \in \BZ$ the wave operators $M_{EE}^2(k)$ and $M_{HH}^2(k)$ are selfadjoint on $L^2_{\eps}(\T^3,\C^3)$ and $L^2_{\mu}(\T^3,\C^3)$. 
		\item $k \mapsto M_{EE}^2(k)$ and $k \mapsto M_{HH}^2(k)$ are analytic on the entire Brillouin zone and equivariant in the sense of equation~\eqref{periodic:eqn:equivariant_operator}. 
		\item The spectra of $M_{EE}^2(k) \, \big \vert_{\ker \, ((\nabla - \ii k) \cdot \eps)}$ and $M_{HH}^2(k) \, \big \vert_{\ker \, ((\nabla - \ii k) \cdot \mu)}$ are purely discrete, \ie they consist of eigenvalues of finite multiplicity which accumulate at $+ \infty$. 
		%
		% \item The spectra of $M_{EE}^2(k) \, \vert_{\mathcal{J}^E(k)}$ and $M_{HH}^2(k) \, \vert_{\mathcal{J}^H(k)}$ are purely discrete, \ie they consist of eigenvalues of finite multiplicity which accumulate at $\infty$.
		%
		\item The essential spectrum consists only of $0$, \ie $\sigma_{\mathrm{ess}} \bigl ( M_{EE}^2(k) \bigr ) = \{ 0 \} = \sigma_{\mathrm{ess}} \bigl ( M_{HH}^2(k) \bigr )$, and is solely due to gradient fields $\ran (- \ii \nabla + k)$. 
		%
		% \item The frequency bands as well as the eigenfunctions that arise from $M_{EE}^2(k)$ and $M_{HH}^2(k)$ are locally analytic away from band crossings.
	\end{enumerate}
\end{lemma}
For the reader's convenience, we have included a proof of these basic facts in Appendix~\ref{appendix:proof_properties_wave_operators}. 
% subsubsection fiber_decomposition_of_the_first_and_second_order_operators (end)

\subsubsection{Equivalence of the frequency band spectra} % (fold)
\label{periodic:frequency_bands:equivalence_spectra}
Thus, the spectra of the first- and second-order operators consist solely of eigenvalues, and we therefore only need to pay attention to the corresponding eigenvalue equations, namely  
\begin{align}
	\Maux(k) \varphi_n(k) = \omega_n(k) \, \varphi_n(k)
	\label{periodic:eqn:eigenvalue_equation_auxiliary_Maxwell_operator}
\end{align}
for the auxiliary Maxwell operator \cite[Theorem~1.4]{DeNittis_Lein:adiabatic_periodic_Maxwell_PsiDO:2013} and 
\begin{subequations}\label{periodic:eqn:eigenvalue_equation_wave_operator}
	\begin{align}
		M_{EE}^2(k) \varphi_n^E(k) = \bigl ( \omega_n^E(k) \bigr )^2 \, \varphi_n^E(k) 
		, 
		\label{periodic:eqn:eigenvalue_equation_wave_operator:electric}
		\\
		M_{HH}^2(k) \varphi_n^H(k) = \bigl ( \omega_n^H(k) \bigr )^2 \, \varphi_n^H(k) 
		, 
		\label{periodic:eqn:eigenvalue_equation_wave_operator:magnetic}
	\end{align}
\end{subequations}
for the two wave operators. 

Once we label these eigenvalues in the obvious way, these give rise to three sets of frequency bands. By convention the label $n = 0$ is reserved for the infinitely degenerate flat band $\omega_0(k) = 0 = \omega_0^{E,H}(k)$ that is due to longitudinal gradient fields. Unlike $M_{EE}^2(k) \geq 0$ and $M_{HH}^2(k) \geq 0$ the auxiliary Maxwell operator $\Maux(k)$ is not bounded from below, and \emph{frequency bands here come as positive-negative frequency pairs} where $\omega_{-n}(k) = - \omega_n(k)$; by convention positive/negative frequency bands are labeled with positive/negative integers. This pairing is due to the symmetry $J = \sigma_3 \otimes \id : (\varphi^E,\varphi^H) \mapsto \bigl (\varphi^E , - \varphi^H \bigr )$, 
\begin{align}
	J \, \Maux(k) \, J^{-1} &= - \Maux(k) 
	. 
	\label{periodic:eqn:chiral_operation}
\end{align}
We emphasize that $J$ is \emph{not} a symmetry of the physical fields because it evidently maps positive onto negative frequency states and vice versa. Therefore, it does not restrict to an operator $\Hil(k) \longrightarrow \Hil(k)$ and is \emph{not a symmetry of the physical system.} For further explanation of this subtle, but very important point, we refer to \cite[Section~3.1]{DeNittis_Lein:symmetries_electromagnetism:2017}. 

The frequency restriction for the Maxwell operator $M(k) = \Maux(k) \, \big \vert_{\omega \geq 0}$ applies fiber-wise, and the relevant first-order eigenvalue equation 
\begin{align}
	M(k) \varphi_n(k) = \omega_n(k) \, \varphi_n(k)
	\label{periodic:eqn:eigenvalue_equation_Maxwell_operator}
\end{align}
is functionally equivalent to \eqref{periodic:eqn:eigenvalue_equation_auxiliary_Maxwell_operator}, we just discard negative frequency solutions. 
% The $k$-dependent positive frequency subspaces $\mathcal{J}_+(k) = \overline{\mathrm{span} \bigl \{ \varphi_n(k) \bigr \}_{n \in \N}}$, which arise from Zak transforming $\mathcal{J}_+ \mapsto \Fourier \mathcal{J}_+ \cong \int_{\BZ}^{\oplus} \dd k \, \mathcal{J}_+(k)$, are spanned by the Bloch functions with positive indices.
\medskip

\noindent
The main aim of this subsection is showing that the non-negative frequency bands $\omega_n(k) = \omega_n^E(k) = \omega_n^H(k)$ coincide. One way to give a constructive proof is by means of the maps $\imath^{E,H}$ from equation~\eqref{setup:eqn:inclusion} that reconstruct the magnetic or electric component. Because $\eps$, $\mu$, $M_{EE}^2$ and $M_{HH}^2$ are periodic, the (periodic) operators $\imath^E \cong \int_{\BZ}^{\oplus} \dd k \, \imath^E(k)$ and  $\imath^H \cong \int_{\BZ}^{\oplus} \dd k \, \imath^H(k)$ fiber decompose into a collection of $k$-dependent operators $\imath^E(k) : \ker \bigl ( (\nabla - \ii k) \cdot \eps \bigr ) \longrightarrow \mathcal{J}_+(k)$ and $\imath^H(k) : \ker \bigl ( (\nabla - \ii k) \cdot \mu \bigr ) \longrightarrow \mathcal{J}_+(k)$ where 
$\mathcal{J}_+(k)$ and the other transversal subspaces are obtained from Zak transforming $\mathcal{J}_+ \mapsto \Fourier \mathcal{J}_+ \cong \int_{\BZ}^{\oplus} \dd k \, \mathcal{J}_+(k)$, $\ker \, ( \nabla \cdot \eps ) \subset L^2_{\eps}(\R^3,\C^3)$ and $\ker \, ( \nabla \cdot \mu ) \subset L^2_{\mu}(\R^3,\C^3)$. These maps allow us to reconstruct $\omega_n(k) > 0$ frequency Bloch functions 
\begin{align}
	\varphi_n(k) = \left (
	\begin{matrix}
		\varphi_n^E(k) \\
		\varphi_n^H(k) \\
	\end{matrix}
	\right )
	= \imath^E(k) \, \varphi_n^E(k)
	= \imath^H(k) \, \varphi_n^H(k)
	&&
	\forall n \in \N
	, \; 
	k \in \BZ
	, 
	\label{periodic:eqn:imath_reconstructs_electromagnetic_Bloch_functions}
\end{align}
of $M(k)$ from the electric and magnetic Bloch functions or compute the magnetic Bloch functions $\varphi_n^H(k)$ from the electric component $\varphi_n^E(k)$. 
\begin{proposition}\label{periodic:prop:reconstruction_eigenfunctions}
	Suppose the material weights satisfy Assumption~\ref{periodic:assumption:periodic_material_weights}. Then all electromagnetic, electric and magnetic frequency bands coincide, 
	\begin{align*}
		\omega_n(k) = \omega_n^E(k) = \omega_n^H(k)
		&&
		\forall n \in \N_0
		, \; 
		k \in \BZ
		. 
	\end{align*}
	More specifically, we have: 
	\begin{enumerate}[(1)]
		\item $\omega_0(k) = 0 = \omega_0^{E,H}(k)$ 
		\item Electric and magnetic parts of the Bloch function $\varphi_n(k) = \bigl ( \varphi_n^E(k) , \varphi_n^H(k) \bigr )$ to \eqref{periodic:eqn:eigenvalue_equation_Maxwell_operator} and $\omega_n(k) > 0$ satisfy equations~\eqref{periodic:eqn:eigenvalue_equation_wave_operator:electric} and \eqref{periodic:eqn:eigenvalue_equation_wave_operator:magnetic}, respectively. 
		\item For $\omega_n(k) > 0$ eigenfunctions to \eqref{periodic:eqn:eigenvalue_equation_wave_operator:electric} and \eqref{periodic:eqn:eigenvalue_equation_wave_operator:magnetic} give rise to eigenfunctions of \eqref{periodic:eqn:eigenvalue_equation_Maxwell_operator} via equation~\eqref{periodic:eqn:imath_reconstructs_electromagnetic_Bloch_functions}. 
		\item The magnetic Bloch function $\varphi_n^H(k) = \mathrm{pr}^H(k) \; \imath^E(k) \, \varphi_n^E(k)$ to $\omega_n(k) > 0$ can be reconstructed from the electric Bloch function $\varphi_n^E(k)$ and vice versa, where 
		\begin{align}
			\mathrm{pr}^{E,H}(k) \, \bigl ( \psi^E(k) , \psi^H(k) \bigr ) := \psi^{E,H}(k) 
			\label{periodic:eqn:projection_pr_E_H}
		\end{align}
		picks out the electric/magnetic component. 
	\end{enumerate}
\end{proposition}
\begin{proof}
	\begin{enumerate}[(1)]
		\item The kernels of the operators $M(k)$, $M_{EE}^2(k)$ and $M_{HH}^2(k)$ consists of gradient fields (\cf Lemma~\ref{periodic:lem:essential_properties_wave_operators}~(4) and \cite[Section~3.2]{DeNittis_Lein:adiabatic_periodic_Maxwell_PsiDO:2013}), and by our labeling convention gives rise to $\omega_0$, $\omega_0^E$ and $\omega_0^H$. 
		\item Evidently, $(\Maux)^2 = M_{EE}^2 \oplus M_{HH}^2$ relates $\Maux$ with the two second-order operators, and $\bigl ( \Maux(k) \bigr )^2 \varphi_n(k) = \bigl ( \omega_n(k) \bigr )^2 \, \varphi_n(k)$ implies $M_{EE}^2(k) \varphi_n^E(k) = \bigl ( \omega_n(k) \bigr )^2 \, \varphi_n^E(k)$ and $M_{HH}^2(k) \varphi_n^H(k) = \bigl ( \omega_n(k) \bigr )^2 \, \varphi_n^H(k)$. 
		\item Suppose $\varphi_n^E(k)$ solves \eqref{periodic:eqn:eigenvalue_equation_wave_operator:electric} for $\omega_n(k) > 0$, then the vector 
		\begin{align*}
			\imath^E(k) \, \varphi_n^E(k) &= \left (
			\begin{matrix}
				\varphi_n^E(k) \\
				- \ii \mu^{-1} \, (\nabla - \ii k)^{\times} \, \bigl ( M_{EE}^2(k) \bigr )^{- \nicefrac{1}{2}} \, \varphi_n^E(k) \\
			\end{matrix}
			\right )
			\\
			&= \left (
			\begin{matrix}
				\varphi_n^E(k) \\
				- \ii \, \omega_n(k)^{-1} \, \mu^{-1} \, (\nabla - \ii k)^{\times} \, \varphi_n^E(k) \\
			\end{matrix}
			\right )
		\end{align*}
		is an eigenvector to $\Maux(k)$ and the eigenvalue $\omega_n(k) > 0$, 
		\begin{align}
			\Maux(k) \, \imath^E(k) \, \varphi_n^E(k) &= \left (
			\begin{matrix}
				\omega_n(k)^{-1} \, \eps^{-1} \, (\nabla - \ii k)^{\times} \, \mu^{-1} \, (\nabla - \ii k)^{\times} \varphi_n^E(k) \\
				- \ii \mu^{-1} \, (\nabla - \ii k)^{\times} \, \varphi_n^E(k) \\
			\end{matrix}
			\right )
			\notag
			\\
			&= \omega_n(k) \, \left (
			\begin{matrix}
				\varphi_n^E(k) \\
				- \ii \, \omega_n(k)^{-1} \, \mu^{-1} \, (\nabla - \ii k)^{\times} \, \varphi_n^E(k) \\
			\end{matrix}
			\right )
			\notag
			\\
			&= \omega_n(k) \; \imath^E(k) \, \varphi_n^E(k) 
			. 
			\label{periodic:eqn:reconstruction_eigenvalue_equation_imath}
		\end{align}
		As the sign of the eigenvalue is positive, we can in fact replace $\Maux(k)$ with $M(k)$ in the above computation. The proof for the magnetic component is completely analogous. 
		\item This follows directly from (2) and (3). 
	\end{enumerate}
\end{proof}
\begin{remark}
	Mathematically, it would be equally possible to give a description in terms of negative frequency bands: the relation~\eqref{periodic:eqn:chiral_operation} implies $J \, \imath^{E,H}(k)$ yields a \emph{negative} frequency Bloch functions. However, note that while both solutions are \emph{mathematically} equivalent, when the material weights $W \neq \overline{W}$ are complex those negative frequency solutions are in fact \emph{unphysical.} 
\end{remark}
The proof of the preceding Proposition shows more, and since we will need this fact later one, we separate it out into a 
\begin{corollary}\label{periodic:cor:imath_pr_are_inverses}
	Suppose the material weights satisfy Assumption~\ref{periodic:assumption:periodic_material_weights}. 
	Then $\imath^{E,H}(k)$ and $\mathrm{pr}^{E,H}(k)$ are inverses to one another, 
	\begin{align*}
		\imath^E(k) \; \mathrm{pr}^E(k) &= \id_{\mathcal{J}_+(k)} 
		, 
		&&
		\mathrm{pr}^E(k) \; \imath^E(k) = \id_{\ker ((\nabla - \ii k) \cdot \eps)} 
		, 
		\\
		\imath^H(k) \; \mathrm{pr}^H(k) &= \id_{\mathcal{J}_+(k)} 
		, 
		&&
		\mathrm{pr}^H(k) \; \imath^H(k) = \id_{\ker ((\nabla - \ii k) \cdot \mu)} 
		, 
	\end{align*}
	and therefore, equation~\eqref{periodic:eqn:imath_reconstructs_electromagnetic_Bloch_functions} holds true. 
\end{corollary}
\begin{proof}
	Proving $\mathrm{pr}^{E}(k) \; \imath^{E}(k) = \id_{\ker ((\nabla - \ii k) \cdot \eps)}$ is immediate as $\imath^E(k)$ leaves the electric component, which $\mathrm{pr}^E(k)$ singles out, untouched. 
		
	The other equality, $\imath^E(k) \; \mathrm{pr}^E(k) = \id_{\mathcal{J}_+(k)}$, requires a bit more work: as the electromagnetic Bloch eigenfunctions for positive frequency bands form a complete basis set of $\mathcal{J}_+(k)$ and the two maps are linear, it suffices to show the statement for electromagnetic Bloch eigenfunctions. But $\varphi_n^H(k) = - \ii \mu^{-1} \, (\nabla - \ii k)^{\times} \, \varphi_n^E(k)$ is precisely what we have shown with the computation~\eqref{periodic:eqn:reconstruction_eigenvalue_equation_imath}. 
	
	The proof for the magnetic field case is identical. 
\end{proof}
\begin{remark}
	It may seem as if the electric or magnetic field component suffices to study the dynamical problem in a more straightforward fashion. After all, the above statement shows that we may reconstruct the magnetic component of any 
	\begin{align*}
		\psi^E(t,k) = \sum_{n \in \N} \e^{- \ii t \omega_n(k)} \, \alpha_n(k) \, \varphi_n^E(k) 
	\end{align*}
	and obtain the electromagnetic field 
	\begin{align*}
		\psi(t,k) = \imath^E(k) \, \psi^E(t,k) 
		= \sum_{n \in \N} \e^{- \ii t \omega_n(k)} \, \alpha_n(k) \, \varphi_n(k) 
		. 
	\end{align*}
	However, the \emph{coefficients} 
	\begin{align*}
		\alpha_n(k) = \bscpro{\varphi_n(k)}{\psi(0,k)}_W = \bscpro{\varphi_n^E(k)}{\psi^E(0,k)}_{\eps} + \bscpro{\varphi_n^H(k)}{\psi^H(0,k)}_{\mu} 
	\end{align*}
	need to be computed in the first-order formalism that requires the electric \emph{and} magnetic field. The linear combination with the coefficients $\alpha_n^E(k) = \snorm{\varphi_n^E(k)}_{\eps}^{-2} \, \bscpro{\varphi_n^E(k)}{\psi^E(0,k)}_{\eps}$ computed only from the electric field is different and does not solve the second-order equation~\eqref{setup:eqn:wave_equation_electric} in Zak representation. 
\end{remark}
%
% subsubsection the_frequency_band_spectra_coincide (end)
% subsection the_frequency_band_spectra_coincide (end)

\subsection{Electromagnetic, electric and magnetic Chern numbers necessarily coincide} % (fold)
\label{periodic:Chern_numbers}
Now that we have proven the equivalence of the dynamical problem (Theorem~\ref{setup:thm:equivalence}) and the frequency band \emph{spectra} (Proposition~\ref{periodic:prop:reconstruction_eigenfunctions}), it seems obvious that also the frequency band \emph{topologies} must coincide. That there is still something to prove might not be obvious, so let us start with that first.

\subsubsection{Setting the stage: why the problem is not solved yet} % (fold)
\label{periodic:Chern_numbers:proof_by_hand}
To strip down the problem to the essentials, we turn our attention to a single, non-degenerate frequency band $\omega_n(k)$ that does not intersect with any other band. Suppose its electromagnetic Bloch function $\varphi_n(k) = \bigl ( \varphi_n^E(k) , \varphi_n^H(k) \bigr )$ is normalized to 
\begin{align}
	1 = \bnorm{\varphi_n(k)}_W^2 = \bnorm{\varphi_n^E(k)}_{\eps}^2 + \bnorm{\varphi_n^H(k)}_{\mu}^2 
	, 
	\label{periodic:eqn:normalization_electromagnetic_Bloch_function}
\end{align}
and its phase is (at least locally) chosen such that $k \mapsto \varphi_n(k)$ is analytic. Note that we do \emph{not} know whether the electric component $\bnorm{\varphi_n^E(k)}_{\eps}^2 = \mathrm{const.}$ and magnetic component $\bnorm{\varphi_n^H(k)}_{\mu}^2 = \mathrm{const.}$ are constant. And we may \emph{not} normalize the electric and magnetic components \emph{separately} to \eg $\nicefrac{1}{2}$ as then we no longer know whether the resulting function is in fact an eigenfunction of \eqref{periodic:eqn:eigenvalue_equation_auxiliary_Maxwell_operator}. The $3 \times 3$ matrix 
\begin{align}
	\mathrm{Ch} := \frac{1}{2 \pi} \int_{\BZ} \dd k \, \Omega(k) 
	\label{periodic:eqn:electromagnetic_Chern_numbers}
\end{align}
that contains the three electromagnetic Chern numbers as its offdiagonal elements is then defined in terms of the \emph{electromagnetic Berry curvature} $\Omega(k) = \bigl ( \Omega_{jl}(k) \bigr )_{1 \leq j,l \leq 3}$, 
\begin{align}
	\Omega_{jl}(k) := \partial_{k_j} \mathcal{A}_l(k) - \partial_{k_l} \mathcal{A}_j(k) 
	= \Omega_{jl}^E(k) + \Omega_{jl}^H(k) 
	, 
	\label{periodic:eqn:electromagnetic_Berry_curvature}
\end{align}
which can be further subdivided into an electric and a magnetic contribution. These are due to the electric and magnetic component of the electromagnetic Berry connection
\begin{align*}
	\mathcal{A}(k) :& \negmedspace = \ii \, \bscpro{\varphi_n(k)}{\nabla_k \varphi_n(k)}_W 
	= \ii \, \bscpro{\varphi_n^E(k)}{\nabla_k \varphi_n^E(k)}_{\eps} + \ii \, \bscpro{\varphi_n^H(k)}{\nabla_k \varphi_n^H(k)}_{\mu} 
	\\
	&=: \mathcal{A}^E(k) + \mathcal{A}^H(k) 
	, 
\end{align*}
although we emphasize that $\mathcal{A}^E$ and $\mathcal{A}^H$ by themselves are \emph{not} connections. That is because we do not know whether the electric and magnetic contributions in the electromagnetic normalization condition~\eqref{periodic:eqn:normalization_electromagnetic_Bloch_function} are constant functions of $k$. 

Instead, we need to define the \emph{electric} Berry connection 
\begin{align*}
	\widetilde{A}(k) :& \negmedspace = \ii \, \bscpro{\widetilde{\varphi}_n^E(k)}{\nabla_k \widetilde{\varphi}_n^E(k)}_{\eps}
	= \norm{\varphi_n^E(k)}_{\eps}^{-2} \, \mathcal{A}^E(k) - \ii \, \nabla_k \bnorm{\varphi_n^E(k)}_{\eps}
\end{align*}
on the basis of a suitably normalized electric Bloch function 
\begin{align*}
	\widetilde{\varphi}_n^E(k) := \norm{\varphi_n^E(k)}_{\eps}^{-1} \, \varphi_n^E(k) 
	. 
\end{align*}
This renormalization is well-defined, because $\norm{\varphi_n^E(k)}_{\eps} \neq 0$ never vanishes thanks to equation~\eqref{periodic:eqn:imath_reconstructs_electromagnetic_Bloch_functions}: because we able to reconstruct the electromagnetic Bloch eigenfunctions solely from the electric part, $\varphi_n^E(k) = 0$ would imply $0 = \imath^E(k) \varphi_n^E(k) = \varphi_n(k)$. But this evidently runs afoul with the electromagnetic normalization condition~\eqref{periodic:eqn:normalization_electromagnetic_Bloch_function}.\footnote{Technically, we only know that $\varphi_n^E(k) \neq 0$ for almost all $k$ at this point. But because $k \mapsto \imath^E(k)$ can be seen to be analytic away from band crossings, this is in fact true for all $k$.} 

Based on the electric Berry connection we may define the electric Berry curvature $\widetilde{\Omega}^E = \bigl ( \widetilde{\Omega}^E_{jl}(k) \bigr )_{1 \leq j , l \leq 3}$, that we can express in terms of the electric and magnetic components of the electromagnetic Berry curvature~\eqref{periodic:eqn:electromagnetic_Berry_curvature}, 
\begin{align*}
	\widetilde{\Omega}_{jl}^E(k) :& \negmedspace = \partial_{k_j} \widetilde{\mathcal{A}}_l^E(k) - \partial_{k_l} \widetilde{\mathcal{A}}_j^E(k) 
	\\
	&= \bnorm{\varphi_n^E(k)}_{\eps}^{-2} \, \Bigl ( \Omega_{jl}^E(k) - \partial_{k_j} \bigl ( \ln \bnorm{\varphi_n^E(k)}_{\eps} \bigr ) \; \mathcal{A}_l^E(k) + \partial_{k_l} \bigl ( \ln \bnorm{\varphi_n^E(k)}_{\eps} \bigr ) \; \mathcal{A}_j^E(k) \Bigr ) 
	, 
	% \label{periodic:eqn:electric_Berry_curvature}
\end{align*}
and the matrix of electric Chern numbers, 
\begin{align}
	\mathrm{Ch}^E := \frac{1}{2 \pi} \int_{\BZ} \dd k \, \widetilde{\Omega}^E(k) 
	. 
	\label{periodic:eqn:electric_Chern_numbers}
\end{align}
The magnetic Berry curvature $\widetilde{\Omega}^H$ and magnetic Chern numbers $\mathrm{Ch}^H$ are defined in the same fashion. 

Comparing equations~\eqref{periodic:eqn:electromagnetic_Chern_numbers} and \eqref{periodic:eqn:electric_Chern_numbers}, we can see no obvious relations between $\mathrm{Ch}$, $\mathrm{Ch}^E$ and $\mathrm{Ch}^H$. In fact, even if we impose simplifying assumptions such as $\snorm{\varphi_n^E(k)}_{\eps} = \mathrm{const.}$, looking at the simplified equations it is still not obvious that 
\begin{align}
	\mathrm{Ch} = \mathrm{Ch}^E = \mathrm{Ch}^H 
	\label{periodic:eqn:Chern_numbers_agree}
\end{align}
all agree. 

To summarize, while it is true that we can reconstruct the electric or magnetic component of Bloch functions with the help of the map $\imath^{E,H}(k)$, the electric and magnetic Chern numbers are \emph{computed} solely from the electric or magnetic fields without reconstructing the other, missing components first. So there is something left to prove. However, a direct, hands-on proof of \eqref{periodic:eqn:Chern_numbers_agree} — even in the simplest situation of a single, non-degenerate band — seems unfeasible. 
% subsubsection attempting_a_proof_by_hand (end)

\subsubsection{The frequency relevant bands} % (fold)
\label{periodic:Chern_numbers:relevant_bands}
Usually more than one band contributes to bulk-edge correspondences, and we will refer to those bands as the \emph{relevant bands}. Just like in solid state physics, the so-called \emph{gap condition} is crucial so as to ensure that the relevant bands decouple from the other bands: 
\begin{assumption}[Relevant Bands and Gap Condition]\label{periodic:assumption:gap_condition}
	Suppose $\specrel(k) := \bigcup_{j \in \mathcal{I}} \bigl \{ \omega_j(k) \bigr \}$, referred to as the \emph{relevant bands}, is a collection of $n$ frequency bands characterized by an index set $\mathcal{I} := \bigl \{ j_1 , \ldots , j_n \bigr \}$. We say that they satisfy the \emph{Gap Condition} if and only if the relevant bands do not cross or merge with other bands. Put mathematically, $0 \not\in \specrel(k)$ and we have 
	\begin{align*}
		\mathrm{dist} \Bigl ( \mbox{$\bigcup_{j \in \mathcal{I}}$} \bigl \{ \omega_j(k) \bigr \} \, , \, \mbox{$\bigcup_{l \not\in \mathcal{I}}$} \bigl \{ \omega_l(k) \bigr \}  \Bigr ) > 0 
		. 
	\end{align*}
\end{assumption}
\begin{remark}
	The condition $0 \not\in \specrel(k)$ is necessary to exclude ground state bands and the longitudinal waves, which have approximately linear dispersion around $k = 0$ and $\omega = 0$. The two positive frequency ground state bands are the only ones that touch $\omega = 0$, and they do so only at the center of the Brillouin zone \cite[Theorem~1.4~(iii)]{DeNittis_Lein:adiabatic_periodic_Maxwell_PsiDO:2013}. At that point they intersect with all of the longitudinal gradient fields, and it is for that reason that the ground state Bloch functions are not even continuous at $k = 0$. Hence, they need to be excluded from the construction below. 
	
	However, to obtain a \emph{physically meaningful} photonic bulk-edge correspondence, we will need to include the ground state bands in our arguments. We shall not attempt to do so here and postpone this to a future work \cite{DeNittis_Lein:photonic_bulk_edge_correspondence:2018}. 
\end{remark}
Suppose we are given relevant bands $\specrel(k) := \bigcup_{j \in \mathcal{I}} \bigl \{ \omega_j(k) \bigr \}$ that satisfy the Gap Condition; the corresponding Bloch functions give rise to the relevant electromagnetic and electric/magnetic subspaces, 
\begin{subequations}\label{periodic:eqn:relevant_subspaces}
	\begin{align}
		\Hil_{\mathrm{rel}}(k) := \mathrm{span} \bigl \{ \varphi_j(k) \bigr \}_{j \in \mathcal{I}}
		, 
		\label{periodic:eqn:relevant_subspaces:electromagnetic}
		\\
		\Hil_{\mathrm{rel}}^{E,H}(k) := \mathrm{span} \bigl \{ \varphi_j^{E,H}(k) \bigr \}_{j \in \mathcal{I}}
		. 
		\label{periodic:eqn:relevant_subspaces:electric}
	\end{align}
\end{subequations}
We can use standard arguments to show that these subspaces depend on $k$ in an analytic fashion: by writing the corresponding orthogonal projections 
\begin{subequations}\label{periodic:eqn:relevant_projections}
	\begin{align}
		P_{\mathrm{rel}}(k) :& \negmedspace = \sum_{j \in \mathcal{I}} \sket{\varphi_j(k)}_W \sbra{\varphi_j(k)}_W
		\label{periodic:eqn:relevant_projections:electromagnetic}
		\\
		P_{\mathrm{rel}}^{E,H}(k) :& \negmedspace = \sum_{j \in \mathcal{I}} \bket{\varphi_j^{E,H}(k)}_{\eps,\mu} \bbra{\varphi_j^{E,H}(k)}_{\eps,\mu}
		\label{periodic:eqn:relevant_projections:electric_or_magnetic}
	\end{align}
\end{subequations}
as a Cauchy integral and exploiting the gap condition, we can transfer the analyticity of the resolvents to the projections and their ranges. 

The bra-ket notation here emphasizes what scalar product to use, \eg we define the rank-$1$ operator $\bket{\psi^E(k)}_{\eps} \bbra{\varphi(k)}_W : L^2_W(\T^3,\C^6) \longrightarrow L^2_{\eps}(\T^3,\C ^3)$ as 
\begin{align*}
	\bket{\psi^E(k)}_{\eps} \bbra{\varphi(k)}_W \phi(k) := \bscpro{\varphi(k)}{\phi(k)}_W \, \psi^E(k) \in L^2_{\eps}(\T^3,\C^3)
	. 
\end{align*}
The matrices of Chern numbers now have straight-forward generalizations to the multiband case that are best expressed in terms of the relevant projections, 
\begin{subequations}
	\begin{align}
		\mathrm{Ch}_{jl} :& \negmedspace = - \frac{\ii}{2 \pi} \int_{\BZ} \dd k \; \mathrm{Tr}_{\Hil(k)} \Bigl ( P_{\mathrm{rel}}(k) \, \bigl [ \partial_{k_j} P_{\mathrm{rel}}(k) \, , \, \partial_{k_l} P_{\mathrm{rel}}(k) \bigr ]  \Bigr ) 
		% = - \frac{\ii}{2 \pi} \int_{\BZ} \dd k \; \mathrm{Tr}_{L^2_W(\T^3,\C^6)} \bigl ( P_{\mathrm{rel}}(k) \, \bigl [ \partial_{k_j} P_{\mathrm{rel}}(k) \, , \, \partial_{k_l} P_{\mathrm{rel}}(k) \bigr ]  \bigr )
		, 
		\\
		\mathrm{Ch}^{E,H}_{jl} :& \negmedspace = - \frac{\ii}{2 \pi} \int_{\BZ} \dd k \; \mathrm{Tr}_{L^2_{\eps,\mu}(\T^3,\C^3)} \Bigl ( P^{E,H}_{\mathrm{rel}}(k) \, \bigl [ \partial_{k_j} P^{E,H}_{\mathrm{rel}}(k) \, , \, \partial_{k_l} P^{E,H}_{\mathrm{rel}}(k) \bigr ] \Bigr ) 
		, 
	\end{align}
\end{subequations}
where $\mathrm{Tr}_{\Hil}$ denotes the trace on the Hilbert space $\Hil = \Hil(k) , L^2_{\eps}(\T^3,\C^3)$ and $j , l = 1 , 2 , 3$. 
% subsubsection the_frequency_relevant_bands (end)

\subsubsection{Construction of the electromagnetic, electric and magnetic Bloch vector bundles} % (fold)
\label{periodic:Chern_numbers:construction_vector_bundles}
One way to view — and, indeed, define — Chern numbers is to view them as topological invariants characterizing complex vector bundles of fixed rank up to isomorphism \cite{Peterson:Chern_classes:1959,Cadek_Vanzura:classification_vector_bundles:1993}. Given a family of relevant bands, in what follows we will associate three vector bundles to them, an electromagnetic, an electric and a magnetic vector bundle. Considering these vector bundles “up to isomorphism” has a neat physical interpretation: topological phases and the topological invariants labeling them do not change under continuous, gap-preserving deformations of the physical system under study. And indeed, the classification of vector bundles up to isomorphism can be obtained from a homotopy definition and the fact that any vector bundle of a given rank can be seen as the pullback of a universal vector bundle \cite[Section~1.2, pp.~27]{Hatcher:vector_bundles_K_theory:2009}. 

The actual definition of the relevant vector bundles is somewhat complicated by us employing the Zak transform where functions and operators are \emph{quasi}-periodic under $\Gamma^*$ translations. One option to remedy this is to switch to the ordinary Bloch-Floquet transform where functions are $\Gamma^*$-periodic in $k$ and $\Gamma$-quasi periodic in $x$ (see \eg \cite[Definition~4.1]{DeNittis_Lein:exponentially_loc_Wannier:2011}). 

However, we will take another route and start with a vector bundle 
\begin{align}
	\bigsqcup_{k \in \R^3} \Hil_{\mathrm{rel}}(k) \overset{\pi}{\longrightarrow} \R^3 
	\label{periodic:eqn:vector_bundle_R3}
\end{align}
over all of $\R^3$, where $\bigsqcup$ denotes the disjoint union and $\pi : \psi(k) \mapsto k$ is the projection onto the base point. Proving that this is indeed a vector bundle is straightforward (see \eg \cite[Lemma~4.5]{DeNittis_Lein:exponentially_loc_Wannier:2011} for the technical arguments that apply verbatim here), and rests on the fact that the projections $k \mapsto P_{\mathrm{rel}}(k)$ are analytic and therefore in particular continuous.\footnote{Thanks to the Oka principle \cite[Remark~4.4]{DeNittis_Lein:exponentially_loc_Wannier:2011} we need not distinguish between analytic and topological vector bundles in the present context.} 

The base space $\R^3$ of momenta comes naturally comes furnished with a $\Gamma^*$ action, and the equivariance~\eqref{periodic:eqn:equivariant_operator} of the projection $k \mapsto P_{\mathrm{rel}}(k)$ implies that the multiplication operator 
\begin{align*}
	\e^{+ \ii \gamma^* \cdot \hat{x}} : \Hil_{\mathrm{rel}}(k) \longrightarrow \Hil_{\mathrm{rel}}(k - \gamma^*)
\end{align*}
relates the fibers at $k$ and $k - \gamma^*$. Moreover, the group action is evidently \emph{free}\footnote{Apart from the identity element $0 \in \Gamma^*$, translations by $\gamma^* \in \Gamma^* \setminus \{ 0 \}$ have no stationary points. }, and therefore \eqref{periodic:eqn:vector_bundle_R3} in fact defines a $\Gamma^*$-equivariant vector bundle. 

And these are naturally isomorphic to what we will call the \emph{electromagnetic Bloch vector bundle} 
\begin{align}
	\mathcal{E}_{\mathrm{rel}} &: \Bigl ( \bigsqcup_{k \in \R^3} \Hil_{\mathrm{rel}}(k) \overset{\pi}{\longrightarrow} \R^3 \Bigr ) \slash \Gamma^* 
\end{align}
over the Brillouin torus $\BZ \simeq \R^3 / \Gamma^* \simeq \T^3$ \cite[Proposition~1.6.1]{Atiyah:K_theory:1994}. Intuitively, this is a fancy way of saying that quasiperiodic objects are completely determined by their behavior over one period and the quasiperiodicity condition. 

Clearly, these ideas also apply to $P_{\mathrm{rel}}^{E,H}$, which gives rise to the electric/magnetic Bloch vector bundles, 
\begin{align*}
	\mathcal{E}_{\mathrm{rel}}^{E,H} &: \Bigl ( \bigsqcup_{k \in \R^3} \Hil_{\mathrm{rel}}^{E,H}(k) \overset{\pi^{E,H}}{\longrightarrow} \R^3 \Bigr ) \slash \Gamma^* 
	. 
\end{align*}
%
% subsubsection construction_of_the_electromagnetic_electric_and_magnetic_bloch_vector_bundles (end)

\subsubsection{Electromagnetic, electric and magnetic Bloch vector bundles are isomorphic} % (fold)
\label{periodic:Chern_numbers:three_vector_bundles_isomorphic}
To streamline the presentation we will focus on the connection between electromagnetic and electric Chern numbers. Because the roles of electric and magnetic fields are symmetric, any and all arguments also apply to the magnetic case. 

The central ingredient here are the linear maps 
\begin{align}
	\imath_{\mathrm{rel}}^E(k) := \imath^E(k) \, \big \vert_{\Hil^E_{\mathrm{rel}}(k)} &: \Hil_{\mathrm{rel}}^E(k) \longrightarrow \Hil_{\mathrm{rel}}(k) 
	\\
	\mathrm{pr}_{\mathrm{rel}}^E(k) := \mathrm{pr}^E(k) \, \big \vert_{\Hil_{\mathrm{rel}}(k)} &: \Hil_{\mathrm{rel}}(k) \longrightarrow \Hil_{\mathrm{rel}}^E(k)
\end{align}
that have been restricted to the relevant subspaces, which are \emph{finite}-dimensional Hilbert spaces. Let us collect some important properties of these restricted maps: 
\begin{lemma}\label{periodic:lem:analyticity_properties}
	Suppose the material weights satisfy Assumption~\ref{periodic:assumption:periodic_material_weights} and we are given a family $\specrel(k) = \bigcup_{j \in \mathcal{I}} \bigl \{ \omega_j(k) \bigr \}$ of relevant bands that satisfy the Gap Condition~\ref{periodic:assumption:gap_condition}. Then the following holds: 
	\begin{enumerate}[(1)]
		\item The orthogonal projections~\eqref{periodic:eqn:relevant_projections} onto the relevant subspaces — and therefore the relevant subspaces~\eqref{periodic:eqn:relevant_subspaces} themselves — depend analytically on $k$ and are equivariant. 
		\item The map $k \mapsto \imath_{\mathrm{rel}}^{E,H}(k) : \Hil_{\mathrm{rel}}^{E,H}(k) \longrightarrow \Hil_{\mathrm{rel}}(k)$ is globally analytic and equivariant in $k$. 
		\item The map $k \mapsto \mathrm{pr}_{\mathrm{rel}}^{E,H}(k) : \Hil_{\mathrm{rel}}(k) \longrightarrow \Hil_{\mathrm{rel}}^{E,H}(k)$ is globally analytic and equivariant in $k$. 
		\item The maps $\imath_{\mathrm{rel}}^{E,H}(k)$ and $\mathrm{pr}_{\mathrm{rel}}^{E,H}(k)$ are inverses to one another. 
	\end{enumerate}
\end{lemma}
\begin{proof}
	\begin{enumerate}[(1)]
		\item A proof was outlined in Section~\ref{periodic:Chern_numbers:relevant_bands}. 
		\item We give the proof only for $\imath^E_{\mathrm{rel}}(k)$ and $\mathrm{pr}^E(k)$, the one for the magnetic maps is identical. First, let us deal with analyticity: from Corollary~\ref{setup:cor:inclusion_bounded_injections} we know that the norm of 
		\begin{align*}
			\imath_{\mathrm{rel}}^E(k) = \left (
			\begin{matrix}
				\id_{\Hil_{\mathrm{rel}}^E(k)} \\
				- \ii \mu^{-1} \, (\nabla - \ii k)^{\times} \, \bigl ( M_{EE}^2(k) \bigr )^{- \nicefrac{1}{2}} \\
			\end{matrix}
			\right )
		\end{align*}
		is $\leq 1$ for almost any $k \in \BZ$; in fact, repeating the relevant steps in the proof of Proposition~\ref{setup:prop:identification_spaces} after replacing $\Maux$ with $\Maux(k)$, we see that this is indeed true for \emph{all} $k$ and not just a subset of full measure.
		
		Moreover, $\imath_{\mathrm{rel}}^E(k)$ acts trivially on the electric component and the electric subspace depends analytically on $k$, we need to place our attention on the magnetic component. 
		
		But this is the product of two analytic operators: the first factor, $\bigl ( M_{EE}^2(k) \bigr )^{- \nicefrac{1}{2}} \, \big \vert_{\Hil_{\mathrm{rel}}^E(k)}$ maps $\Hil_{\mathrm{rel}}^E(k)$ to itself. And since the operator $M_{EE}^2(k) \, \big \vert_{\Hil_{\mathrm{rel}}^E(k)} \geq \bigl ( \inf_{k \in \BZ} \specrel(k) \bigr ) \; \id_{\Hil_{\mathrm{rel}}^E(k)}$ is bounded away from $0$, we may use the Taylor expansion of $(1 + x)^{- \nicefrac{1}{2}}$ to find a local power series of $\bigl ( M_{EE}^2(k) \bigr )^{- \nicefrac{1}{2}} \big \vert_{\Hil_{\mathrm{rel}}^E(k)}$ in $k - k_0$ near any fixed $k_0$ in the Brillouin zone $\BZ$. As the maximal frequency $\sup_{k \in \BZ} \specrel(k) < \infty$ is bounded, we know that the radius of convergence is always strictly positive and uniformly bounded away from $0$ independently of the point $k_0$ around which we expand. The other operator, $- \ii \, \mu^{-1} \, (\nabla - \ii k)^{\times}$, is linear and defined on a space that changes analytically, and hence, analytic. 
		
		Putting this together, we have shown that $k \mapsto \imath_{\mathrm{rel}}^E(k)$ is analytic on the entire Brillouin zone. Equivariance follows from writing this map as $\imath_{\mathrm{rel}}^E(k) = \sum_{j \in \mathcal{I}} \sket{\varphi_n(k)}_W \bra{\varphi_n^E(k)}_{\eps}$ and using the equivariance of the Bloch functions themselves. 
		\item The proof is simple: the Hilbert spaces the map is defined on depend analytically on $k$ by (1), and the actual prescription is independent of $k$. Hence, the map $\mathrm{pr}_{\mathrm{rel}}^E$ is analytic. Moreover, the equivariance of $\mathrm{pr}_{\mathrm{rel}}^E(k) = \sum_{j \in \mathcal{I}} \sket{\varphi_n^E(k)}_{\eps} \sbra{\varphi_n(k)}_W$ can be read off by expressing it in the corresponding eigenbases. 
		\item First of all, the Gap Condition guarantees $\omega_n(k) > 0$ holds for all $k$ and $n \in \mathcal{I}$. Hence, as a consequence of Corollary~\ref{periodic:cor:imath_pr_are_inverses} the Hilbert spaces $\Hil_{\mathrm{rel}}(k)$ and $\Hil_{\mathrm{rel}}^E(k)$ both have dimension $\abs{\mathcal{I}}$, and also the restrictions are inverses to one another. 
	\end{enumerate}
\end{proof}
These properties of the maps $k \mapsto \imath_{\mathrm{rel}}^{E,H}(k)$ and its inverse $k \mapsto \mathrm{pr}_{\mathrm{rel}}^{E,H}(k)$ can now be restated in vector bundle theoretic language, using \cite[Theorem~2.5, p.~27]{Husemoller:fiber_bundles:1966}, as follows: 
\begin{proposition}\label{periodic:prop:equivalence_vector_bundles}
	Suppose the material weights satisfy Assumption~\ref{periodic:assumption:periodic_material_weights} and the relevant bands fulfill the Gap Condition~\ref{periodic:assumption:gap_condition}. Then $k \mapsto \imath_{\mathrm{rel}}^{E,H}(k)$ and $k \mapsto \mathrm{pr}_{\mathrm{rel}}^{E,H}(k)$ define vector bundle isomorphisms that are inverse to each other, and therefore $\mathcal{E}_{\mathrm{rel}} \simeq \mathcal{E}^{E,H}_{\mathrm{rel}}$ are isomorphic. 
\end{proposition}
And because analytic vector bundles over the three-dimensional torus are uniquely characterized by three Chern numbers, we immediately get a proof for Theorem~\ref{intro:thm:Chern_numbers_agree}, that can be restated as 
\begin{corollary}
	Therefore, all three Chern numbers of the electromagnetic Bloch vector bundle, the electric Bloch vector bundle and the magnetic Bloch vector bundle agree. 
\end{corollary}
%
% subsubsection  (end)
% section equivalence_of_frequency_band_pictures_in_periodic_media (end)

%% file: appendix.tex
%!TEX root = /Users/max/Dropbox/research/photonic crystals/active/equivalence first- and second-order formalism/equivalence first- and second-order formalism electromagnetism.tex
%
\begin{appendix}
	\section{Properties of $M_{EE}^2(k)$ and $M_{HH}^2(k)$} % (fold)
	\label{appendix:proof_properties_wave_operators}
	\begin{proof}[Lemma~\ref{periodic:lem:essential_properties_wave_operators}]
		Just like before, we will only spell out some of the proofs for the electric field operator $M_{EE}^2$, everything applies to $M_{HH}^2$ after exchanging the roles of $\eps$ and $\mu$. 
		\begin{enumerate}[(1)]
			\item This follows immediately from the selfadjointness of $M_{EE}^2$ and $M_{HH}^2$. 
			\item Because of the $\Gamma$-periodicity of $\eps$ and $\mu$, the fiber operators necessarily satisfy the equivariance condition~\eqref{periodic:eqn:equivariant_operator}. 
		
			To show analyticity, we note that $M_{EE}^2(k)$ is a quadratic polynomial in $k$, consisting of $M_{EE}^2(0)$ and lower-order operators which are infinitesimally small compared to it. Hence, we deduce from the Kato-Rellich theorem \cite[Theorem~X.12]{Reed_Simon:M_cap_Phi_2:1975} that the domain $\domain \bigl ( M_{EE}^2(k) \bigr ) = \domain \bigl ( M_{EE}^2(0) \bigr )$ is independent of $k$ as well. Combining these two factoids yields analyticity of $k \mapsto M_{EE}^2(k)$ for all $k \in \R^3$. 
			%
			% CHANGED Make sure that \mathcal{J}(k) and \mathcal{J}^{E,H}(k) have been defined. Otherwise define them here. 
			\item While we could show this directly, starting from the definitions of $M_{EE}^2$ and $M_{HH}^2$, we will give a proof that exploits the connection to $\Maux$ and what has already been proven about the auxiliary Maxwell operator in the literature. 
			
			To make the notation more compact, we abbreviate the divergence-free electromagnetic, electric and magnetic subspaces with
			\begin{align*}
				\mathcal{J}(k) :& \negmedspace = \ker \, \bigl ( \Div(k) \, W \bigr )
				\subset L^2_W(\T^3,\C^6)
				, 
				\\
				\mathcal{J}^E(k) :& \negmedspace = \ker \, \bigl ( (\nabla - \ii k) \cdot \eps \bigr )
				\subset L^2_{\eps}(\T^3,\C^3)
				, 
				\\
				\mathcal{J}^H(k) :& \negmedspace = \ker \, \bigl ( (\nabla - \ii k) \cdot \mu \bigr )
				\subset L^2_{\mu}(\T^3,\C^3)
				, 
			\end{align*}
			where $\Div(k) \bigl ( \psi^E , \psi^H \bigr ) := \bigl ( (\nabla - \ii k) \, \cdot \, \psi^E \, , \, (\nabla - \ii k) \, \cdot \, \psi^H \bigr )$ is the electromagnetic divergence in the fiber labeled with Bloch momentum $k$. 
			
			According to \cite[Theorem~3.4]{DeNittis_Lein:adiabatic_periodic_Maxwell_PsiDO:2013} the resolvent $\bigl ( \Maux(k) \, \vert_{\mathcal{J}(k)} - z \bigr )^{-1}$ is compact when $z \in \C \setminus \R$ as the spectrum of the selfadjoint operator $\Maux(k)$ is necessarily a subset of the reals. Then $z^2$ cannot be an element of $\sigma \bigl ( \Maux(k)^2 \bigr )$ and we therefore deduce that also 
			\begin{align*}
				\bigl ( \Maux(k)^2 \, \vert_{\mathcal{J}(k)} - z^2 \bigr )^{-1} &= \bigl ( \Maux(k) \, \vert_{\mathcal{J}(k)} - z \bigr )^{-1} \; \bigl ( \Maux(k) \, \vert_{\mathcal{J}(k)} + z \bigr )^{-1}
				\\
				&= \bigl ( M_{EE}^2(k) \, \vert_{\mathcal{J}^E(k)} - z^2 \bigr )^{-1} \oplus \bigl ( M_{HH}^2(k) \, \vert_{\mathcal{J}^H(k)} - z^2 \bigr )^{-1}
			\end{align*}
			is compact. Here, the direct sum refers to the decomposition of 
			\begin{align*}
				L^2_W(\T^3,\C^6) = L^2_{\eps}(\T^3,\C^3) \oplus L^2_{\mu}(\T^3,\C^3) 
			\end{align*}
			into electric and magnetic subspaces. This is evidently the case if and only if both, $\bigl ( M_{EE}^2(k) \, \vert_{\mathcal{J}^E(k)} - z^2 \bigr )^{-1}$ and $\bigl ( M_{HH}^2(k) \, \vert_{\mathcal{J}^H(k)} - z^2 \bigr )^{-1}$ are compact. 
		
			Hence, the non-negative operator $M_{EE}^2(k) \, \vert_{\mathcal{J}^E(k)}$ restricted to the transversal subspace and its magnetic counterpart $M_{HH}^2(k) \, \vert_{\mathcal{J}^H(k)}$ have purely discrete spectra, consisting solely of eigenvalues of finite multiplicity accumulating at $\infty$. 
			\item Proposition~\ref{setup:prop:identification_spaces}~(1) translates fiber-wise: $(\nabla - \ii k) \times (\nabla - \ii k) \varphi = 0$ holds for all $\varphi \in \Cont^{\infty}_{\mathrm{per}}(\T^3)$, and therefore gradient fields (which form an infinite-dimensional subspace) contribute to the essential spectrum at $0$. In fact, revisiting the arguments in the proof Proposition~\ref{setup:prop:identification_spaces}~(1), we deduce that the kernel of $M_{EE}^2(k)$ consists \emph{only} of gradient fields. 
		
			To show that the essential spectra consist only of $0$, we note that due to the Helmholtz-type decomposition 
			\begin{align*}
				L^2_{\eps}(\T^3,\C^6) = \ran \, (- \ii \nabla + k) \oplus \bigl ( \ran \, (- \ii \nabla + k) \bigr )^{\perp_{\eps}}
				= \ran \, (- \ii \nabla + k) \oplus \mathcal{J}^E(k) 
			\end{align*}
			for the electric field and the fact that the spectrum of the wave operator $M_{EE}^2(k) \, \vert_{\mathcal{J}^E(k)}$ restricted to the transversal subspace $\mathcal{J}^E(k)$ is purely discrete (point~(3)). That shows the claim. 
			%
			% \item Due to the equivariance of $\Maux(k)$ and $M(k)$ frequency band functions are $\Gamma^*$-periodic. And standard arguments using the Riesz formula to express projections onto the eigenspace associated to a single band $\omega_n(k)$ as a contour integral yields that away from band crossings, $k \mapsto \omega_n(k)$ and the associated eigenspace depend \emph{analytically} on $k$. Moreover, if we choose the neighborhood of a point $k_0$ small enough, we may also choose an analytic basis of that subspace. This concludes the proof of the Lemma.
		\end{enumerate}
	\end{proof}
	%
	% section properties_of_m__ee_2_k_and_m__hh_2_k (end)
\end{appendix}
%